\documentclass[letterpaper, conference]{IEEEtran}
\IEEEoverridecommandlockouts
\usepackage{cite}
\usepackage{amsmath,amssymb,amsfonts}
\usepackage{algorithmic}
\usepackage{graphicx}
\usepackage{textcomp}
\usepackage{xcolor}
\def\BibTeX{{\rm B\kern-.05em{\sc i\kern-.025em b}\kern-.08em
		T\kern-.1667em\lower.7ex\hbox{E}\kern-.125emX}}

\usepackage{booktabs} 

\usepackage{wrapfig} 
\usepackage{floatflt}
\usepackage{amsthm}
\usepackage{amsmath}
\usepackage{amsfonts}
\usepackage{verbatim}
\usepackage{tikz}
\usepackage{multirow}
\usepackage{graphicx}
\usepackage{caption}
\usepackage{subcaption}
\usepackage{amssymb}
\usepackage{tcolorbox}

\usepackage{xcolor}
\usepackage{lipsum}
\usepackage{subfiles}
\usepackage{adjustbox}

\usepackage{algorithm}
\usepackage{algorithmic}
\usepackage{array}
\usepackage{makecell}
\usepackage{slashbox}

\newtheorem{lemma}{Lemma}
\newtheorem{theorem}{Theorem}
\newtheorem{definition}{Definition}

\newcommand{\dep}{$\ \textbf{dep}\ $}
\newcommand{\K}{\mathcal{K}}

\begin{document}

	\title{Global Stabilization for Causally Consistent Partial Replication
		\thanks{This research is supported in part by National Science Foundation award 1849599, and Toyota InfoTechnology Center. Any opinions, findings, and conclusions or recommendations expressed here are those of the authors and do not necessarily reflect the views of the funding agencies or the U.S. government.}
	}
	
	\author{\IEEEauthorblockN{Zhuolun Xiang}
		\IEEEauthorblockA{\textit{Department of Computer Science} \\
			\textit{University of Illinois at Urbana-Champaign}\\
			xiangzl@illinois.edu}
		\and
		\IEEEauthorblockN{Nitin H. Vaidya}
		\IEEEauthorblockA{\textit{Department of Computer Science} \\
			\textit{Georgetown University}\\
			nitin.vaidya@georgetown.edu}
	}
	
	\maketitle
	
	\begin{abstract}
		Causally consistent distributed storage systems have received significant attention recently due to the potential for providing high throughput and causality guarantees. 
		{\em Global stabilization} is a technique established for achieving causal consistency in distributed multi-version key-value store systems, adopted by the previous work such as GentleRain \cite{Du2014GentleRainCA} and Cure \cite{akkoorath2016cure}. Intuitively, this approach serializes all updates by their physical time and computes the ``Global Stable Time'' which is a time point $t$ such that versions with timestamp $\leq t$ can be returned to the client without violating causality. However, all previous designs with global stabilization assume {\em full replication}, where each data center stores a full copy of data, and each client is restricted to access servers within one data center. In this paper, we propose a theoretical framework to support {\em general partial replication} with  causal consistency via global stabilization, where each server can store an arbitrary subset of the data, and each client is allowed to communicate with any subset of the servers and migrate among them without extra delays. 
		We propose an algorithm that implements causal consistency for distributed multi-version key-value stores with general partially replication.
		We prove the optimality of the Global Stable Time computation in our algorithm regarding the remote update visibility latency, i.e. how fast update from a remote server is visible to the client, under general partial replication.
		We also provide trade-offs to further optimize the remote update visibility by introducing extra delays during client's migration. 
		Simulation results on the performance of our algorithm compared to the previous work are also provided. 
	\end{abstract}
	
	
	\section{Introduction}\label{sec:intro}
	The purpose of this paper is to propose {\em global stabilization} for implementing causal consistency in a {\em partially replicated distributed storage system}. Geo-replicated storage system plays a vital role in many distributed systems, providing fault-tolerance and low latency when accessing data. In general, there are two types of replication methods, {\em full replication} where the same set of data are replicated at each server or data center, and {\em partial replication} where each server can store a different subset of the data. As the amount of data stored grows rapidly, partial replication is receiving an increasing attention \cite{dahlin2006practi, Hlary2006AboutTE, Shen2015CausalCF, crain2015designing, xiang2017lower, bravo2017saturn}.
	
	To simplify the applications developed based on distributed storage, many systems provide consistency guarantees when clients access the data. Among various consistency models, causal consistency has received significant attention recently, for its emerging applications in social networks.  To ensure causal consistency, when a client can get a version of some key, it must be able to get versions of other keys that are causally preceding. 
	
	There have been numerous designs for causally consistent distributed storage systems, especially in the context of full replication. For instance, {Lazy Replication} \cite{Ladin1992ProvidingHA} and {SwiftCloud} \cite{Zawirski2014SwiftCloudFG} utilize vector timestamps as metadata for recording and checking causal dependencies. {COPS} \cite{Lloyd2011DontSF} and {Bolt-on CC} \cite{bailis2013bolt} keep dependent updates  explicitly to maintain the causality. 
	{GentleRain} \cite{Du2014GentleRainCA} proposed the global stabilization technique for achieving causal consistency, which trades off throughput with data freshness. 
	{Eunomia} \cite{gunawardhana2017unobtrusive} also uses global stabilization but only within each data center, and serializes updates between data centers in a total order that is consistent with causality. 
	{Occult} \cite{mehdi2017can} moves the dependency checking to the read operation issued by the client to prevent data centers from cascading. 
	
	In terms of partial replication, there is some recent progress as well. 
	{PRACTI} \cite{dahlin2006practi} implements a protocol that sends updates only to the servers that store the corresponding keys, but the metadata is still sent to all servers.
	In contrast, our algorithm only requires sending metadata to a necessary subset of servers.
	{Saturn} \cite{bravo2017saturn} implements tree-based metadata dissemination via a shared tree among the datacenters to provide both high throughput and data visibility. All updates between data centers are serialized and transmitted through the shared tree. 
	Our algorithm does not require to maintain such shared tree topology for propagating metadata. Instead, our algorithm allows updates and metadata from one server to be sent to another server directly, without the extra cost of maintaining a shared tree topology among the servers.

	Most relevant to this paper is the {\em global stabilization} techniques used in GentleRain \cite{Du2014GentleRainCA}. Distributed systems often require its components to exchange heartbeat messages periodically in order to achieve fault tolerance. In the design of GentleRain, each server is equipped with a loosely synchronized physical clock for acquiring the physical time.
	When sending heartbeats, the value of physical clock is piggybacked with the message.  
	Also, the timestamp for each update message is the physical time when the update is issued, and all updates are serialized in a total order by their timestamps.
	The communication between any two servers is via a FIFO channel, hence the timestamp received by one server from another server is always monotonically increasing. 
	Suppose the latest timestamp server $i$ receives from server $j$ is $t$, then any updates from $j$ to $i$ with timestamp $\leq t$ has already been received by server $i$.
	Due to the total ordering of all updates by their physical time, to achieve causal consistency, each server $i$ only need to calculate the time point $T$ such that the latest timestamp value received from any other server is no less than $T$. This indicates that server $i$ has received all updates with timestamp $\leq T$ from other servers, and hence there will be no causal dependency missing if server $i$ returns versions with timestamp $\leq T$.
	We call such time point $T$ as the {\em Global Stable Time} or $GST$. 
	
	However, there are several constraints on the design of GentleRain. 
	In particular, 
	(i) GentleRain applies to only full replication, where each datacenter stores a full copy of all the data (key-value pairs).
	Within a data center, the key space is partitioned among the servers in that data center, and such partition needs to be identical for every data center, (ii) each client can only access servers within one data center. Under these constraints, the global stabilization approach is simple and straightforward.
	
	In this paper, we develop a theoretical framework for {\em general partial replication} via global stabilization where (i) we allow {\em arbitrary data replication} across all the servers, and (ii) each client can communicate with an {\em arbitrary subset of servers} for accessing data, and migrate among the servers {\em without extra delays}.
	As we will see in Section \ref{sec:compute_gst}, the global stabilization technique, which is relatively simple in the case of full replication, becomes much more complicated under general partial replication, due to the arbitrary data sharing pattern and clients' mobility.
	Finding the right way to compute the optimal Global Stable Time for general partial replication is the main challenge of this paper.
	
	The contributions of this paper are the following:
	\begin{enumerate}
		\item We propose an algorithm that implements causal consistency for general partially replicated distributed storage system. The algorithm allows each server to store an arbitrary subset of the data, and each client can communicate with an arbitrary subset of the servers and migrate among them without extra delays.
		\item We prove the optimality of the $GST$ computation in our algorithm regarding remote update visibility latency, i.e., how fast update from the remote server is visible to the client, under general partial replication.
		\item We also provide trade-offs to further optimize the remote update visibility latency by introducing extra delays during client's migration.
		\item We provide simulation results on the performance of our algorithm comparing to the stabilization algorithm of GentleRain.
	\end{enumerate}

	\section{System Model}
	\label{sec:model}
	
	\begin{floatingfigure}[r]{0.5\linewidth}
		\includegraphics[width=0.5\linewidth]{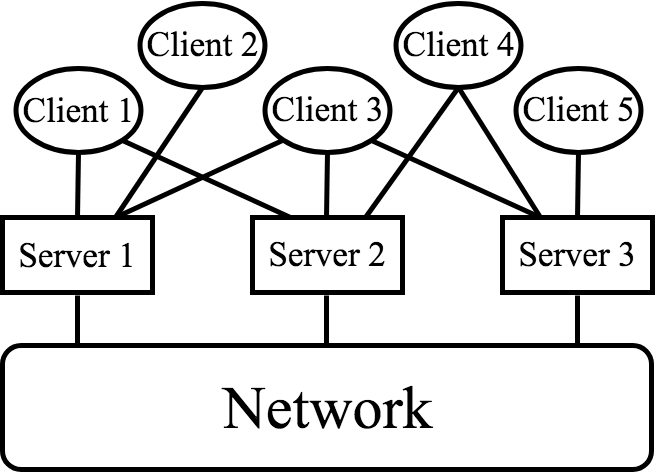}
		\caption{Illustration of the system model}
		\label{fig:model}
	\end{floatingfigure}
	
	We consider a client-server architecture,  as illustrated in Figure \ref{fig:model}. 
	Let there be $n$ servers, $\mathcal{S}=\{1,\cdots\, n \}$. Let there be $m$ clients, $\mathcal{C}=\{1,\cdots, m\}$. 
	Each client $c$ is restricted to communicate with an arbitrary set of servers $S_{c}$, and we will call $S_c$ the server set of client $c$.
	We assume that client $c$ can access all the keys stored at any server in $S_{c}$.
	Let $\mathcal{G}$ be the set containing all clients' server sets, i.e. $\mathcal{G}=\{S_{c}~|~\forall \text{ client }c\}$.
	Notice that the size of $\mathcal{G}$ is $|\mathcal{G}|\leq 2^n$ where $n$ is the total number of servers.
	We say a client migrates from server $i$ to server $j$, if the client issues some operation to server $i$ first, and then to server $j$.

	The communication channel between servers is assumed to be {\em point-to-point, reliable and FIFO}. 
	Each server has multi-version key-value storage locally, where a new version of a key is created when a client writes a new value to that key. Each version of a key also stores some metadata for the purpose of maintaining causal consistency. 
	Each server has a physical clock (reflects the physical time in the real world) that is loosely synchronized across all servers by some time synchronization protocol such as NTP \cite{ntp}. 
	Each server will periodically send heartbeat messages (denote as HB) with its physical clock value to a selected subset of servers (the choice of the subset is described later).
	The clock synchronization precision may only affect the performance of our algorithm, not the correctness.
	
	To access the data, a client can issue GET(key) and PUT(key, value) to a server. GET(key) will return to the client with the value of the key as well as some metadata. PUT(key, value) will create a new {\em version} of the key at the server, and return to the client with some metadata.
	We call all PUT operations to some server $i$ as {\em local PUT} at $i$, and all other PUT operations as {\em non-local PUT} with respect to $i$.
	
	\subsection{Model for General Partial Replication}
	We allow {\em arbitrary} replication of the keys among the servers, i.e. each server can store an arbitrary subset of the keys. 
	Let $\K_{i}$ denote the set of keys stored at server $i$.
	Let $\K_{ij}=\K_{i}\cap \K_{j}$ denote the set of keys shared by servers $i$ and $j$.
	For example, in Figure \ref{fig:asg}, let $\mathcal{K}_i=\{k',k,y,a\}$, $\mathcal{K}_{1}=\{k',x,b\}$, $\mathcal{K}_{j}=\{v, d\}$, then $\mathcal{K}_{i1}=\{k'\}$, $\mathcal{K}_{ij}=\emptyset$.
	
	In order to model the data partition, 
	we define a {\em share graph}, which was originally introduced by H{\'e}lary and Milani \cite{Hlary2006AboutTE}. We also define a {\em augmented share graph} that further captures how clients access servers. 
	
	\begin{definition}[Share Graph \cite{Hlary2006AboutTE}]\label{def:sg}
		Share graph is an unweighted undirected graph, defined as $G^s=(V^s,E^s)$, where $V^s=\{1,2,\cdots,n\}$, where
		vertex $i\in V^s$ represents server $i$, and there exists an undirected edge $(i,j)\in E^s$ if $\mathcal{K}_{ij}\neq \emptyset$.
	\end{definition}
	
	The augmented share graph extends the share graph by adding virtual edges between nodes $i,j$ such that $i,j\in S_c$ for some client $c$.
	
	\begin{definition}[Augmented Share Graph \cite{xiang2017lower}]\label{def:asg}
		Augmented share graph is an unweighted undirected multi-graph, defined as $G^a=(V^a,E^a)$. $V^a=\{1,2,\cdots,n\}$, where
		vertex $i\in V^a$ represents server $i$. There exists a {\em real} edge $(i,j)\in E^a$ if $\mathcal{K}_{ij}\neq \emptyset$, and there exists a {\em virtual} edge $(i,j)\in E^a$ if there exists some client $c$ such that $i,j \in S_c$. Denote the set of real edges in $G^a$ as $E_1(G^a)$ and the set of virtual edges in $G^a$ as $E_2(G^a)$.
	\end{definition}

	{\bf Example:} Figure \ref{fig:asg} shows an example of the augmented share graph defined above. In the example, $G^a$ consists $7$ vertices $h,i,j,1,2,3,4$, and the common keys shared by any two servers are labeled on each edge. There exists a client $c$ that can access $h,i,j$, thus vertices $h,i,j$ are connected by virtual edges. 
	\begin{floatingfigure}[r]{0.5\linewidth}
		\includegraphics[width=0.5\linewidth]{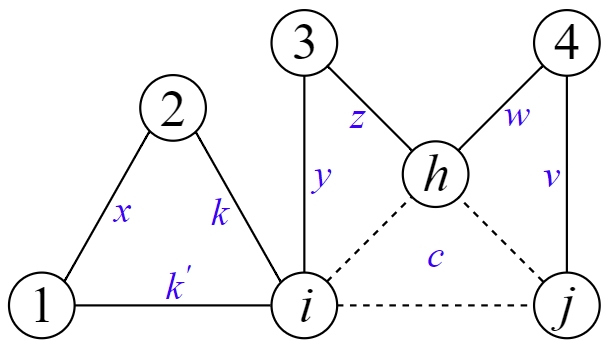}
		\caption{Illustration of $G^a$}
		\label{fig:asg}
	\end{floatingfigure}
	For convenience, we assume that both $G^s$ and $G^a$ are connected. However, our results can be easily extended to the case when the graph is partitioned. 
	We assume the augmented share graph is {\em static} for most of the paper, and briefly discuss how our algorithm may be adapted when there is data insertion/deletion or adding/removing servers in Section \ref{sec:dynamic}.
	
	\subsection{Causal Consistency}
	
	Now we provide the formal definition of causal consistency.
	Firstly, we define the happened-before relation for a pair of operations.
	
	\begin{definition} [Happened-before \cite{lamport1978time}]
		\label{def:hb}
		Let $e$ and $f$ be two operations ($PUT$ or $GET$).  $e$ happens before $f$, denoted as $e \rightarrow f$, if and only if at least one of the following rules is satisfied: 
		\begin{enumerate}
			\item $e$ and $f$ are two operations by the same client, and $e$ happens earlier than $f$
			\item $e$ is a PUT($k,v$) operation, $f$ is a GET($k$) operation and GET($k$) returns the value written by $e$
			\item there is another operation $g$ such that $e \rightarrow g$ and $g \rightarrow f$. 
		\end{enumerate}
	\end{definition}
	
	The above happens-before relation defines a standard causal relationship between two operations. Recall that each client's PUT operation will create a new {\em version} of the key.
	
	\begin{definition} [Causal Dependency \cite{roohitavaf2017causalspartan}]
		Let $K$ be a version of key $k$, and $K'$ be a version of key $k'$. We say $K$ causally depends on $K'$, and denote it as $K \dep K'$ if and only if PUT($k',K'$) $\rightarrow$ PUT($k,K$). 
		We use $\lnot (K\dep K')$ to denote that $K$ does not causally depend on $K'$.
	\end{definition}
	
	Now we define the meaning of visibility for a client.
	
	\begin{definition} [Visibility \cite{roohitavaf2017causalspartan}]
		A version $K$ of key $k$ is visible to a client $c$, if and only if $GET(k)$ issued by client $c$ to any server in $S_c$ returns a version $K'$ such that $K'=K$ or $\lnot (K\dep K')$. We say $K$ is visible to a client $c$ from a server $i$ if the version $K$ is returned from server $i$.
	\end{definition}
	
	We say a client $c$ can access a key $k$ if the client can issue PUT and GET operations to a server that stores $k$.
	Causal consistency is defined based on the visibility of versions to the clients as follows.
	
	\begin{definition} [Causal Consistency \cite{roohitavaf2017causalspartan}]\label{def:cc}
		The key-value storage is causally consistent if both of the following conditions are satisfied.
		\begin{itemize}
			\item Let $k$ and $k'$ be any two keys in the store. Let $K$ be a version of key $k$, and $K'$ be a version of key $k'$ such that $K \dep K'$. 
			For any client $c$ that can access both $k$ and $k'$, when $K$ is read by client $c$, $K'$ is visible to $c$.
			
			\item Version $K$ of a key $k$ is visible to a client $c$ after $c$ completes PUT($k,K$) operation.
		\end{itemize}
		
	\end{definition}

	In Section \ref{sec:algo}, we will first present the structure of the algorithm for both clients and servers.
	Then in Section \ref{sec:compute_gst}, we complete the algorithm by specifying the definition of the {\em Heartbeat Summary (HS)} and {\em Global Stable Time (GST)} used for maintaining causal consistency.
	We also prove in Section \ref{sec:opt} the optimality of our algorithm regarding remote update visibility latency, i.e., how fast update is visible to clients at remote servers, under general partial replication.
	By introducing extra delays during client's migration, we present algorithms in Section \ref{sec:optimize}  that can provide a trade-off between the visibility latency and client migration latencies.
	The evaluation of our algorithm is provided in Section \ref{sec:evaluation}.
	More discussions can be found in Section \ref{sec:discuss}.

	\section{Algorithm}
	\label{sec:algo}
	In this section, we propose the algorithms for the client (Algorithm \ref{alg:client}) and the server (Algorithm \ref{alg:server}).
	The algorithm structure is inspired by GentleRain \cite{Du2014GentleRainCA} and designed for general partial replication. 
	The main idea of our algorithm is to {\em serialize} all PUT operations and resulting versions by their {\em physical clock time} (which is a scalar). 
	For all causally dependent versions, our algorithm guarantees that the total order established by their timestamps is consistent with their causal relation, i.e., if $K \dep K'$ then $K$'s timestamp is strictly larger than $K'$'s timestamp.
	Such ordering simplifies causality checking since now each server can learn that up to which physical time point it has received updates from other servers when assuming FIFO channels between all servers. 
	When a server returns a version $K$ of key $k$ to a client, the server needs to guarantee that all causally dependent versions of $K$ are already visible to the client. How to decide the version of the key to returning is the main challenge of our algorithm, as represented by computing and using Global Stable Time ($GST$) in the algorithm below and Section \ref{sec:compute_gst}.
	While $GST$ is relatively easy to compute for full replication as in GentleRain, we will show that general partial replication makes the computation of optimal $GST$ much more complicated.
	
	When presenting our algorithm in this section, we left the Global Stable Time ($GST$) and Heartbeat Summary ($HS$) undefined, and the definitions are provided later in Section \ref{sec:compute_gst}. 
	Intuitively, $GST$  defines a time point, and the versions no later than this time point can be returned to the client while satisfying causal consistency. $HS$ is a component for computing $GST$.
	We prove the correctness of our algorithm in Section \ref{sec:proof}.
	We also prove in Section \ref{sec:opt} that our definition of $GST$ is optimal regarding the remote update visibility latency, i.e., how fast a version of a remote update is visible to the client. 
	In Table \ref{symbols} below, we provide a summary of the symbols used in our algorithm.
	Recall that $S_c$ is the set of servers that client $c$ can access, and $\mathcal{G}=\{S_{c}~|~\forall \text{ client }c\}$.
	
	\begin{table}[htp]
		\centering
		\begin{tabular}{|l|l|}
			\hline
			Symbols    & Explanations \\ \hline
			$ut$    & update time, scalar \\
			$K$    & version of some key $k$ with value $v$, tuple $< k,v,ut>$ \\  \hline
			$GT_c$    & metadata stored at client $c$ for get dependencies, scalar \\ 
			$PT_c$    & metadata stored at client $c$ for put dependencies, scalar \\ 
			$HS_c$    & Heartbeat Summary stored at client $c$, vector of size $|S_c|$ \\  \hline
			$HS_i(g)$    & Heartbeat Summary for server set $g\in \mathcal{G}$ at server $i$, scalar \\
			$GST$    & Global Stable Time, scalar \\ \hline
			$g$    & server set that $g\in \mathcal{G}$ \\ 
			$N_i^s$    & set of neighbors of server $i$ in the share graph excluding $i$ \\ 
			$HB_{ji}$    & heartbeat value from server $j$ to server $i$ \\ 
			$Clock_i$    & physical clock at server $i$ \\ 
			$O_i$    & set of servers that server $i$ needs to send heartbeat to\\ 
			\hline
			
		\end{tabular}
		\caption{Explanations of symbols}
		\label{symbols}
	\end{table}
	
	Algorithm \ref{alg:client} is the client's algorithm. Each client is restricted to issue GET and PUT operations to the servers in $S_c$. Each client will store a put dependency clock $PT_c$ (which is a scalar) for PUT operations, a get dependency clock $GT_c$ (scalar) for GET operations, and a vector $HS_c$ of length $|S_c|$ for remote dependencies. All these parameters will be specified in Section \ref{sec:compute_gst}. When issuing operations, the client will attach its clocks with the operation, as in lines $3,9$ in Algorithm \ref{alg:client}. When receiving the result from the server, the client will update its clocks as in lines $5,6,11$ in Algorithm \ref{alg:client}.
	
	\begin{algorithm} 
		{
			\caption{Client operations at client $c$.}
			\label{alg:client}
			\begin{algorithmic} [1]
				
				\STATE \textbf{GET}(key $k$) from server $i$
				\STATE \hspace{3mm}  compute $rd(c,i)=\min_{j\in S_c, j\neq i} HS_c[j]$
				\STATE \hspace{3mm}  send $\langle \textsc{GetReq} \ k,PT_c, rd(c,i), S_c\rangle$  to server $i$
				\STATE \hspace{3mm}  receive $\langle \textsc{GetReply} \ v, t,  \{hs_j~|~j\in S_c,j\neq i\}\rangle$ from server $i$
				\STATE \hspace{3mm}  $GT_c\leftarrow \max(GT_c, t)$
				\STATE \hspace{3mm}  $HS_c[j] \leftarrow max (HS_c[j], hs_j)$ for all $j\in S_c, j\neq i$
				\RETURN $v$
				
				\STATE \vspace{5mm} \textbf{PUT}(key $k$, value $v$) to server $i$
				\STATE \hspace{3mm} send $\langle \textsc{PutReq} \ k,v, \max(PT_c, GT_c) \rangle$ to server $i$
				\STATE \hspace{3mm} receive $\langle \textsc{PutReply} \ t \rangle$ 
				\STATE \hspace{3mm} $PT_c \leftarrow \max(PT_c, t)$
				
			\end{algorithmic}
		}
	\end{algorithm}
	
	\begin{algorithm}
		{
			\caption{Server operations at server $i$}
			\label{alg:server}
			\begin{algorithmic} [1]
				\STATE \textbf{upon} receive $\langle \textsc{GetReq} \ k, t, rd, g\rangle$ from client $c$
				\STATE \hspace{3mm} // The computation of $GST$ is provided in Section \ref{sec:compute_gst}
				\STATE \hspace{3mm} \textbf{if} $k$ shared by $j\in g\cap N_i^s$ \textbf{then}
				\STATE \hspace{6mm} \textbf{wait until} $GST\geq t$
				\STATE \hspace{3mm} obtain the latest version $K$ of key $k$ with largest timestamps from local storage s.t. 
				$K.ut \leq GST$ or $K$ is due to a local PUT operation at server $i$
				\STATE \hspace{3mm} send $\langle\textsc{GetReply} \  K.v, K.ut, \{HS_j(g)~|~j\in g, j\neq i\}\rangle$ to client $c$

				\STATE \vspace{5mm} \textbf{upon}  receive $\langle \textsc{PutReq} \ k, v, t \rangle$ from client $c$
				\STATE \hspace{3mm}  \textbf{wait until} $t<Clock_i$
				\STATE \hspace{3mm}  create new version $K$
				\STATE \hspace{3mm}  $K.k \leftarrow k$, $K.v \leftarrow v$, $K.ut \leftarrow Clock_i$
				\STATE \hspace{3mm}  insert $K$ to local storage
				\STATE \hspace{3mm}  \textbf{for} each server $j$ that stores key $k$ \textbf{do}
				\STATE \hspace{6mm}  send $\langle \textsc{Update} \ u_{K}=K \rangle$ to $j$
				\STATE \hspace{3mm}  send $\langle \textsc{PutReply} \ K.ut \rangle$ to client $c$

				\STATE \vspace{5mm} \textbf{upon}  receive $\langle \textsc{Update} \  u\rangle$ from $j$
				\STATE \hspace{3mm}  insert $u$ to local storage
				\STATE \hspace{3mm} $HB_{ji} \leftarrow u.ut$

				\STATE \vspace{5mm} \textbf{upon} every $\Delta$ time
				\STATE \hspace{3mm}  \textbf{for} each server $j\in O_i$ \textbf{do}
				\STATE \hspace{6mm}  send $\langle \textsc{Heartbeat } Clock_i \rangle$ to $j$

				\STATE \vspace{5mm} \textbf{upon}  receive $\langle \textsc{Heartbeat } hb \rangle$ from $j$
				\STATE \hspace{3mm}  $HB_{ji} \leftarrow hb$

				\STATE \vspace{5mm} \textbf{upon}  every $\theta$ time
				\STATE \hspace{3mm} compute $HS_i(g)$ for every $g\in \mathcal{G}$ such that  $i\in g$
				\STATE \hspace{3mm} \textbf{for} each server $j\in g$ \textbf{do}
				\STATE \hspace{6mm}  send $\langle \textsc{Heartbeat Summary } HS_i(g),  g\rangle$ to $j$
				
				\STATE \vspace{5mm} \textbf{upon}  receive $\langle \textsc{Heartbeat Summary } hs, g \rangle$ from $j$
				\STATE \hspace{3mm}  $HS_j(g) \leftarrow hs$
				
			\end{algorithmic}
		}
	\end{algorithm}
	
	Algorithm \ref{alg:server} below is inspired by the algorithm in \cite{Du2014GentleRainCA}, with several important differences: 
	(1) 
	The Global Stable Time computation is different and more complicated due to the general partial replication, as will be specified in Section \ref{sec:compute_gst}.
	(2) 
	The heartbeat/HS exchange procedures are different (lines $19-20$, $25-26$ in Algorithm \ref{alg:server}).
	(3)
	The client will keep slightly more metadata locally, such as a vector of length $|S_c|$.
	(4) 
	There may be blocking for the GET operation of the client as in lines $3,4$ of Algorithm \ref{alg:server}. Such blocking is necessary for satisfying the second condition of causal consistency as in Definition \ref{def:cc}, i.e., the version of client's own PUT is always visible to the client.
	
	The intuition of the algorithm is straightforward. When handling GET operations, the server will first check if the client may have issued a PUT at other servers on some key that it also stores, and make sure such version is visible to the client (lines $3,4$). Then the server will return the latest version of the key that satisfies causal consistency (line $5$). The computation of Global Stable Time (GST) is designed for this purpose, as will be specified in Section \ref{sec:compute_gst}.
	When handling PUT operations, the server will first wait until its physical clock exceeds the client's causal dependencies (line $8$). Then the server performs a put locally (lines $9,10,11$), sends the update to other servers that stores the same key (lines $12,13$), and replies to the client (line $14$). 
	
	Lines $15-17$ is for receiving updates from other servers. Rest of the algorithm (lines $18-28$) specifies how heartbeats and HSs are exchanged among the servers.

	\section{Computing Global Stable Time}
	\label{sec:compute_gst}
	
	In this section, we complete the algorithm by defining heartbeat exchange procedure and Global Stable Time computation. We will specify for each server the set of destination servers its heartbeat/HS messages need to be sent to and how to compute $GST$ from received messages. The Global Stable Time is a function of the augmented share graph defined in Section \ref{sec:model}. 
	As we will see in this section and Section \ref{sec:opt}, the computation of the optimal $GST$ is much more complicated than GentleRain due to general partial replication.

	\subsection{Server Side: $GST$ Computation and Heartbeat Exchange}
	
	Let $HB_{xy}$ denote the clock value attached with the heartbeat message sent from server $x$ to $y$. We will later use the term heartbeat value, heartbeat message or heartbeat to refer $HB_{xy}$.
	Basically, the Global Stable Time ($GST$) in our Algorithm \ref{alg:server} computes a time point that is ``safe'' for returning versions whose timestamps are no larger than this time point. More specifically, $GST$ is computed as {\em the minimum of a set of heartbeat values}, which is the time point that all the causal dependencies have been received at corresponding servers.
	In this section, we provide the computation of $GST$.
	
	We say a cycle or path is simple if it has no vertex repetition.
	We define the length of a cycle to be the number of nodes in the cycle. Nodes $a,b$ with both a real edge and a virtual edge between $a,b$ is considered a valid simple cycle of length $2$. We will use $(a,b)$ to denote the {\bf directed edge} from node $a$ to $b$. We will next define two sets $L_i(k)$ and $R_i(g)$ each contains a set of directed edges.
	
	Define set $L_i(k)$ with respect to server $i$ and a key $k\in \mathcal{K}_i$ as follows. For every simple cycle $(i,v_1,\cdots,v_m, i)$ of length $\geq2$ in $G^a$ such that $m\geq 1$, $k\in \mathcal{K}_{v_1i}$, we have $(v_1, i)\in L_i(k)$, and  if $(v_m,i)$ is a real edge, we also have $(v_m,i)\in L_i(k)$. 
	For instance, in Figure \ref{fig:intuition}, $L_i(k)=\{(1,i),(2,i)\}$.
	Intuitively, if $(v,i)\in L_i(k)$,  then server $v$ may send updates to $i$ that are causal dependencies of key $k$'s version. 
	For example, there can be updates $u_{K'}\rightarrow u_X\rightarrow u_K$, as shown in Figure  \ref{fig:intuition}.

	\begin{floatingfigure}[r]{0.6\linewidth}
		\includegraphics[width=0.6\linewidth]{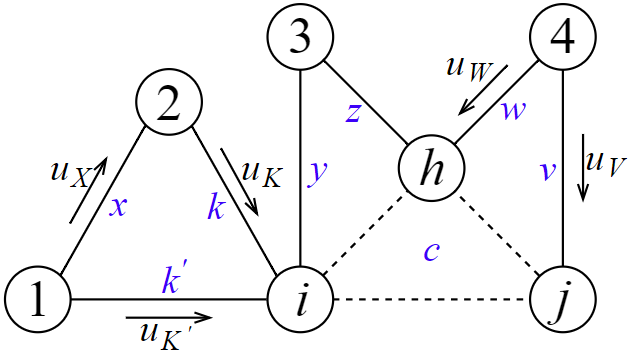}
		\caption{Intuition for set $L_i(k),R_i(g)$}
		\label{fig:intuition}
	\end{floatingfigure}
	
	Recall that  $\mathcal{G}$ is the set of all clients' server sets, i.e. $\mathcal{G}=\{S_c~|~\forall \text{ client }c\}$, and $|\mathcal{G}|\leq 2^n$ where $n$ is the total number of servers. 
	
	Define set $R_i(g)$ with respect to server $i\in g$ and $g\in \mathcal{G}$ as follows. For every simple path $(v_1,\cdots,v_m)$ in $G^a$ such that $v_1,v_m\in g$, $m\geq 2$,  we have $(v_2, v_1)\in R_i(g)$ if $v_1\neq i$ and $(v_2, v_1)$ is a real edge.
	For instance, in Figure \ref{fig:intuition}, let $g=S_c=\{h,i,j\}$, then $R_i(g)=\{(3,h),(4,h),(4,j)\}$.
	Intuitively, if $(a,b)\in R_i(g)$, then server $a$ may send updates to $b$  that are causal dependencies of key $k$'s version. For example in Figure \ref{fig:intuition}, there can be updates $u_V\rightarrow u_W$, and then some client $c'$ reads version $W$ from server $h$ and puts a new version $K$ of key $k$ to server $i$, leading to $K\dep V$.

	As mentioned, set $L_i(k)\cup R_i(g)$ contains directed edges along which the causal dependencies of key $k$'s version may be sent, and these dependencies can be read by client $c$ whose server set is $S_c=g$.
	The computation of $GST$ involves all heartbeat values in the set 
	\[
	\{HB_{xy}~|~ (x,y)\in L_i(k)\cup R_i(g)\}
	\]
	To be more specific, the following two values need to be computed for $GST$: \[LD_i(k)=\min_{(v,i)\in L_i(k)}(HB_{vi}), \,\,\,
	RD_i(g)=\min_{(x,y)\in  R_i(g)}(HB_{xy})
	\]
	which stands for local dependencies (LD) and remote dependencies (RD) respectively.
	The intuition for $LD_i(k)$ is to compute the time point up to which server $i$ has received all causally dependent updates of key $k$'s version. For example in Figure \ref{fig:intuition}, suppose $u_{K'}.ut=0$, $u_X.ut=\epsilon$ and $u_K.ut=2\epsilon$ where $\epsilon$ is some small number. Our algorithm guarantees that if updates $u\rightarrow v$, then $u.ut<v.ut$ as will be shown in the next section.  Recall that servers communicate via FIFO channels, once server $i$ has received $HB_{1i}\geq 2\epsilon$ and $HB_{2i}\geq 2\epsilon$, it has received all the causal dependencies of version $K$ from its neighbors in the augmented share graph. Therefore for version $K$ or similarly other versions of $k$ with timestamp $\leq LD_i(k)$, server $i$ has received the causal dependencies of those versions from its neighbors.
	The intuition for $R_i(g)$ is similar, which computes the time point when all servers in the server set $g$ have received all the causal dependencies of key $k$'s version. More details can be found in the correctness proof of our algorithm in the next section.
	
	\textsl{Heartbeat and HS exchange.}
	
	In order to compute $LD_i(k)$, server $i$ needs to know the set of heartbeat values $HB_{vi}$ for all pairs $(v,i)\in L_i(k)$. Therefore,
	\begin{itemize}
		\item For $\forall v$ such that $(v,i)\in L_i(k)$, $v$ will send heartbeat messages to $i$.
	\end{itemize}
	In order to compute $RD_i(g)$, server $i\in g$ needs to know the value of $\min_{(v,j)\in R_i(g)} HB_{vj}$ for each server $j$ such that $\exists(v,j)\in R_i(g)$. Therefore,
	\begin{itemize}
		\item For $\forall v,j$ such that $(v,j)\in R_i(g)$, $v$ will send heartbeat messages to $j$. Notice that $j\neq i$ by the definition of $R_i(g)$. 
		\item For each server $j$ above, $j$ will periodically send to $i$ a summary of heartbeats (denoted as Heartbeat Summary or HS) it received, as \[HS_j^i(g)=\min_{ (v,j)\in R_i(g)}(HB_{vj})\]
		Note that if $(v,j)\in R_i(g)$ then $j\in g$.
		Also notice that for $\forall i,i'\in g$ and $j\neq i,i'$, by definition $HS_j^i(g)=HS_j^{i'}(g)$, since the set $\{(v,j)\in R_i(g)\}=\{(v,j)\in R_{i'}(g)\}$. We will denote $HS_j(g)=HS_j^i(g)$ for brevity. 
	\end{itemize}
	Then $RD_i(g)=\min_{j\in g, j\neq i}(HS_j(g))$ by the definition of HS above.
	The target server set $O_i$ that server $i$ needs to send heartbeats to can be written as $O_i=\{j|(i,j)\in L_j(k), k\in \mathcal{K}_{j} \}\cup \{j|(i,j)\in R_z(g), z\in \mathcal{S}, g\in \mathcal{G} \}$.
	
	Finally, the computation of $GST$ used in our Algorithm also depends on the client's dependency clock $rd$.
	Intuitively, due to the delay of communication between servers, the values of $HS$s  may be different at different servers in $g$. For instance, server $i$ may receive $HS_j(g)=10$ from server $j$ at time $t$, but server $i'$ may only receive an old message $HS_j(g)=5$ at $t$ due to network delay.
	To avoid such inconsistency, the client $c$ accessing server set $g$ will keep the value of the largest $HS_j(g)$ it has seen so far for $\forall j\in g$, denoted as $HS_c[j]$. 
	And the client's dependency clock $rd(c,i)$ is defined as 
	\[
	rd(c,i)=\min_{j\in S_c, j\neq i} HS_c[j]
	\]
	Since client's dependency clock $rd(c,i)$ (or $rd$) reflects latest remote dependencies that have been observed by the client, when computing $GST$, the larger value between $RD_i(g)$ and $rd$ should be considered for remote dependencies.
	Therefore, the computation of $GST$ can be written as 
	\begin{equation*}
	\begin{aligned}
	GST=\min\left( LD_i(k), \max(RD_i(g), rd)\right)
	\end{aligned}
	\end{equation*}
	\subsection{Client Side}\label{sec:gst_client}
	
	Each client maintains a vector of size $|g|=|S_c|$ for $HS$ values as mentioned above.
	Also, the client will keep two scalars $GT_c$ and $PT_c$ as the dependency clock for GET and PUT dependencies respectively.

	\section{Correctness of Algorithm \ref{alg:client} and \ref{alg:server}}\label{sec:proof}
	In this section, we prove that our Algorithm \ref{alg:client} and \ref{alg:server} implement causal consistency by Definition \ref{def:cc}.

	\begin{lemma}\label{lem:1}
		Suppose that PUT($k',K'$) $\rightarrow$ PUT($k,K$), and thus $K\dep K'$. Let $u_{K'},u_{K}$ denote the corresponding updates of PUT($k',K'$) and PUT($k,K$), and let $u_{K'}.ut,u_{K}.ut$ denote their timestamps. Then $u_{K'}.ut < u_{K}.ut$, and ${K'}.ut<K.ut$.
	\end{lemma}

	\begin{proof}
		The proof is provided in Appendix \ref{sec:lemma1proof}.
	\end{proof}

	\begin{lemma}\label{lem:2}
		Suppose at some real time $t$, a version $K$ of key $k$ is read by client $c$ from server $i$. 
		Consider any server $i'\in S_c$ and version $K'$ of key $k'\in \K_{i'}$ such that $K'$ is due to a PUT at some server other than $i'$, and $K\dep K'$.
		Then at time $t$, 
		(i) $K'$ has been received by server $i'$,
		(ii) the version $K'$ is visible to client $c$ from server $i'$.
	\end{lemma}

	\begin{proof}
		The proof is provided in Appendix \ref{sec:lemma2proof}.
	\end{proof}

	\begin{theorem}\label{thm:1}
		The key-value storage is causally consistent.
	\end{theorem}
	
	\begin{proof}
		The proof is provided in Appendix \ref{sec:thmproof}.
		
	\end{proof}

	\section{Optimality of the Algorithm}\label{sec:opt}
	In this section, we prove that the $GST$ computed by our algorithm is optimal for general partial replication regarding remote update visibility latency,
	which is defined as {\em the period from when a remote update is received by the server to when the remote update is visible to the client.} 
	Recall that in general partial replication, clients are allowed to migrate among the servers freely without extra delays, and our $GST$ is optimal for this case. Later in Section \ref{sec:optimize}, we show that if extra delays can be introduced during the client's migration, the remote update visibility latency can be further reduced.
	To show the optimality for general partial replication, we show that at line $5$ of Algorithm \ref{alg:server}, returning any version with a timestamp larger than our $GST$ value may violate causal consistency, indicating our definition of $GST$ is optimal regarding remote update visibility latency.
	Formally, we have the following theorem. 
	\begin{theorem}\label{thm:2}
		Consider Algorithm \ref{alg:client} and \ref{alg:general} for general partial replication.
		If any version $K$ with $K.ut>GST$ is returned to client $c$ from server $i$ as a result of its $GET(k)$ operation, the causal consistency may be violated. More specifically, there may exists a version $K'$ of some key $k'$ such that $K\dep K'$ and client $c$ can access key $k'$, but version $K'$ is not visible to client $c$.
	\end{theorem}
	
	\begin{proof}
		The proof is provided in Appendix \ref{sec:optproof}.
	\end{proof}

	\section{Optimization for Better Visibility}\label{sec:optimize}
	Previously in Section \ref{sec:algo} and \ref{sec:compute_gst}, we allow each client to migrate among the servers in $S_c$ without extra delays.
	In reality, the frequency of such migration may be low, i.e. a client is likely to communicate with a single server for a long period before changing to another one.
	If such migration among different servers occurs infrequently, 
	it is reasonable to introduce extra delays during the migration, in exchange for better remote update visibility latency when clients issue GET operations.
	In fact, some system designs already observed such trade-off, such as Saturn \cite{bravo2017saturn}. However, Saturn's solution requires to maintain an extra shared tree topology among all the servers, and is quite different from our global stabilization approach.
	In Section \ref{sec:basic_migrate} below, we demonstrate how to design the algorithm to achieve better remote update visibility latency as the discussion above. Then in Section \ref{sec:general_migrate}, we generalize the above idea from a single server to a group of servers.
	
	\subsection{One Server as a Group}\label{sec:basic_migrate}
	We will use the same notation from Section \ref{sec:algo} and \ref{sec:compute_gst}.
	Recall that the Global Stable Time $GST$, computed for the client $c$ accessing server $i$ for the value of key $k$, is the minimum of a set of heartbeat clock values, reflecting all possible local dependencies and remote dependencies. Essentially, the reason for taking remote heartbeat values received by servers other than $i$ is to ensure that the client can migrate freely among the servers in its server set $S_c$. During the client's migration to another server, there is no extra delay since all causal dependencies are guaranteed to be visible to the client as proven in Lemma \ref{lem:2}. 
	One natural idea is that, if the client can wait for a certain period during its migration to ensure that the client's causal dependencies are visible from the target server, then the $GST$ computation does not need to include the remote heartbeat values necessarily.
	To be more specific, the Global Stable Time simply becomes
	\[
	GST = LD_i(k)=\min_{(v,i)\in L_i(k)}(HB_{vi})
	\]
	which only reflects the causal dependencies locally.
	
	When a client migrates to another server, it needs to execute operation \textbf{MIGRATE} as shown in Algorithm \ref{alg:migrate}. 
	Basically, the client will send its dependencies clock $\max(PT_c, GT_c)$ to the new target server for migration. 
	For the target server, it needs to ensure the local storage has already included all the versions in the client's causal dependencies before returning an acknowledgment. Specifically, the server needs to wait until $\min_{k\in \mathcal{K}_i}(LD_i(k))$ is no less than the client's dependency clock, as shown in line $13$ of Algorithm \ref{alg:migrate}.

	\begin{algorithm} 
		{
			\caption{One Server as a Group}
			\label{alg:migrate}
			\begin{algorithmic} [1]
				\STATE // Client operations at client $c$
				\STATE \textbf{MIGRATE} to server $i$
				\STATE \hspace{3mm}  send $\langle \textsc{Migrate} \ \max(PT_c, GT_c)\rangle$  to server $i$
				\STATE \hspace{3mm}  wait for $\langle \textsc{Reply} \rangle$
				\RETURN

				
				\STATE \vspace{5mm}// Server operations at server $i$
				\STATE \textbf{upon} receive $\langle \textsc{Migrate} \ t\rangle$ from client $c$
				\STATE \hspace{3mm} \textbf{wait until} $t\leq \min_{k\in \mathcal{K}_i}(LD_i(k))$
				\STATE \hspace{3mm} send $\langle\textsc{Reply}  \rangle$ to client
				

			\end{algorithmic}
		}
	\end{algorithm}
	
	Also, there is no exchange of Heartbeat Summary among the servers, since now the computation of $GST$ does not dependent on the remote heartbeat values.
	This implies significant savings in bandwidth usage as the number of servers increases.
	
	Another advantage of Algorithm \ref{alg:migrate} is to decrease the visibility latency.
	As mentioned, the $GST$  is now equal to $LD_i(k)$, which is very likely to be larger than the original GST, because the original $GST$ also takes the remote heartbeat values for computation.
	Therefore the version returned is likely to have larger timestamps and thus fresher compared to Algorithm \ref{alg:server}. 
	Although there are extra delays incurred during the client's migration procedure as in line $13$ of Algorithm \ref{alg:migrate}, the penalty caused by migration delays is small if the frequency of migration is low.
	
	\subsection{Multiple Servers as a Group}\label{sec:general_migrate}
	In the basic case, we consider a single server as a ``group'', and introduce extra delays when clients migrate from one group to another.
	In general, a client may frequently access some subset of servers for some time, and then migrate to another subset of servers for frequent accessing. 
	For instance, each subset may be a data center that consists of several servers, and each client usually accesses only one datacenter for PUT/GET operations.
	In this case, each ``group'' that the client will access contains a subset of servers.

	\begin{algorithm} 
		{
			\caption{Multiple Servers as a Group}
			\label{alg:general}
			\begin{algorithmic} [1]
				\STATE // Client operations at client $c$
				\STATE \textbf{MIGRATE} to another group $g'$
				\STATE \hspace{3mm}  send $\langle \textsc{Migrate} \ \max(PT_c, GT_c), g'\rangle$  to some server $i\in g'$
				\STATE \hspace{3mm}  receive $\langle \textsc{Reply } \{hs_j\} \rangle$ from server $i$
				\STATE \hspace{3mm}  $HS_c[j] \leftarrow max (HS_c[j], hs_j)$ for all $j\in g'$
				\RETURN

				\STATE \vspace{5mm} \textbf{GET}(key $k$) from server $i\in g$
				\STATE \hspace{3mm}  compute $rd=\min_{j\in g, j\neq i} HS_c[j]$
				\STATE \hspace{3mm}  send $\langle \textsc{GetReq} \ k,PT_c, rd, g\rangle$  to server $i$
				\STATE \hspace{3mm}  receive $\langle \textsc{GetReply} \ v, t,  \{hs_j\}\rangle$ from server $i$
				\STATE \hspace{3mm}  $GT_c\leftarrow \max(GT_c, t)$
				\STATE \hspace{3mm}  $HS_c[j] \leftarrow max (HS_c[j], hs_j)$ for all $j\in g$
				\RETURN $v$
				
				\STATE \vspace{5mm}// Server operations at server $i$
				\STATE \textbf{upon} receive $\langle \textsc{Migrate} \ t, g\rangle$ from client $c$
				\STATE \hspace{3mm} \textbf{wait until} $t\leq \min(\min_{k\in \mathcal{K}_i}(LD_i(k)), RD_i(g))$
				\STATE \hspace{3mm} send $\langle\textsc{Reply }\{HS_j(g)~|~j\in g\}  \rangle$ to client
				
				\STATE \vspace{5mm} \textbf{upon} receive $\langle \textsc{GetReq} \ k, t, rd, g\rangle$ from client $c$
				\STATE \hspace{3mm} // $GST=\min(LD_i(k), \max(RD_i(g), rd))$
				\STATE \hspace{3mm} \textbf{if} $k$ shared by $j\in g\cap N_i^s$ \textbf{then}
				\STATE \hspace{6mm} \textbf{wait until} $t\leq GST$
				\STATE \hspace{3mm} obtain latest version $K$ of key $k$ with largest timestamps from local storage s.t. 
				$K.ut \leq GST$ or $K$ is due to a local PUT operation at server $i$
				\STATE \hspace{3mm} send $\langle\textsc{GetReply} \  K.v, K.ut, \{HS_j(g)~|~j\in g\}\rangle$ to client
			\end{algorithmic}
		}
	\end{algorithm}
	
	Thus,
	we can design an algorithm where the client can migrate among the servers within a group without extra delays, and need to wait extra time when migrating across different groups, as presented in Algorithm \ref{alg:general}. We only show the different parts compared to the algorithm in Section \ref{sec:algo} here for brevity. 
	
	We will use the same notation from Section \ref{sec:algo} and \ref{sec:compute_gst}.
	The augmented share graph in this section contains virtual edges connecting all servers accessible by one client, including servers within the same group and across groups.
	Then, when a client is accessing group $g$, and issues GET operation to server $i$, the Global Stable Time is computed as 
	\[
	GST=\min(LD_i(k), \max( RD_i(g), rd))
	\]
	where $rd=\min_{j\in g, j\neq i} HS_c[j]$, $HS_c$ is the vector of Heartbeat Summarys stored at client $c$.
	Note for the case $g=\{i\}$, by definition $GST=LD_i(k)$ since $R_i(g)=\emptyset$.
	
	When the client migrates to another group $g'$, extra delay will be enforced. In particular, the server $i'$ in group $g'$ needs to wait until $\min(\min_{k\in \mathcal{K}_{i'}}(LD_{i'}(k)), RD_{i'}(g))\geq t$, where $t$ is the dependency clock of the client. The extra delay here ensures that all client's causal dependencies has been received by the servers in the group $g'$, and visible to the client. 
	
	Notice that the algorithm in Section \ref{sec:algo} and Algorithm \ref{alg:migrate} are both special cases of Algorithm \ref{alg:general}, where group $g$ equals $S_c$ and some single server $i$ respectively.

	\section{Simulation Results}\label{sec:evaluation}
	In this section, we evaluate the heartbeat message overhead and the remote update visibility latency (or visibility latency in short) of our algorithm comparing to the global stabilization algorithm of GentleRain (or GentleRain in short) \cite{Du2014GentleRainCA}. 
	Some simulation results are deferred to the Appendix \ref{sec:experiment} due to lack of space.
	The remote update visibility latency is defined as {\em the period from when a remote update is received by the server to when this remote update is visible to the client}.
	
	Recall in Section \ref{sec:opt}, we have proved that our $GST$ computation is optimal in terms of remote update visibility latency for general partial replication. To give some insights on how well our algorithm performs, we provide simulation results on remote update visibility latency under various settings.
	
	\subsection{Simulation Setup}
	For evaluation purpose, we implement and evaluate the global stabilization layer as described in our algorithm from Section \ref{sec:algo}. We simulate servers by running multiple server processes within a single machine, and control network latencies by manually adding extra delays to all network packages.  
	Each server process will execute multiple threads concurrently, including i) one thread that periodically sends heartbeat messages to target server processes according to the heartbeat frequency ii) one thread that periodically sends update messages (due to $PUT$ operations) to target nodes according to the update throughput iii) one thread that listens and receives messages from other nodes and iv) one thread that periodically computes $GST$ and checks which remote updates are visible. 
	We use synthetic workloads for the simulation.
	The machine used in this experiment runs Ubuntu 16.04 with 8-core CPU of 3.4GHZ, 16 GB memory and 128GB SSD storage. The program is written in {\em Golang}, and uses standard TCP socket communication for exchanging messages.
	
	We evaluate our algorithm for a family of share graphs for the ease of comprehension. The graphs used are {\em ring graphs} of size $n$, with each node to be both a client and a server. The client of one node will only access the server of that node. This family of share graph can represent simple robotic networks in practice -- each node is a robot that stores key-value pairs depending on its physical location, and only share keys with its neighbors. In order to achieve causal consistency, by our algorithm, each node will send heartbeat messages to only its neighbors, and $GST$ is computed as the minimum of the heartbeat values received from its neighbors. As for the global stabilization algorithm in GentleRain, they cannot handle partial replication directly. Therefore we  pretend the system to be fully replicated so that GentleRain can achieve causal consistency correctly. Then, in GentleRain, the $GST$ for each node is computed as the minimum of heartbeat values from all nodes in the ring. 
	Hence intuitively, GentleRain will have a smaller $GST$ value comparing to our algorithm because its $GST$ is computed as the minimum of a larger set of heartbeat values. This implies that only older versions can be visible to the client comparing to our algorithm, which leads to higher remote update visibility latencies. Also, the heartbeat message overhead should be larger in GentleRain.

	In each experiment, we repeat the measurement $3$ times and take the average as a data point. Each experiment will vary one or two parameters while keeping other parameters constant. The default parameters for all experiments are listed below: stabilization frequency = $1000/\sec$,
	heartbeat frequency = $10/\sec$,  
	network delay = $0ms$ or $100ms$, 
	ring size = $10$, 
	update throughput = $5k/\sec$ and 
	clock skew = $0ms$.
	
	\subsection{Simulation Results and Observations}

	\paragraph*{\bf Message Overhead}
	We first measure the overhead of heartbeat messages in our algorithm and GentleRain, as a function of the ring size. Here the heartbeat frequency is set to be $50/\sec$. 
	The overhead presented below is computed as the {\em average overhead over all the servers}.
	As we can see from Figure \ref{fig:overhead}, the message cost is almost constant in our algorithm, while the cost increases dramatically in GentleRain. It is because our algorithm only requires each server to receive heartbeat messages from a small set of servers (neighbors in the ring) in order to achieve causal consistency, while GentleRain needs heartbeat messages from all other servers.
	
    \begin{floatingfigure}[r]{0.52\linewidth}
		\includegraphics[width=0.52\linewidth]{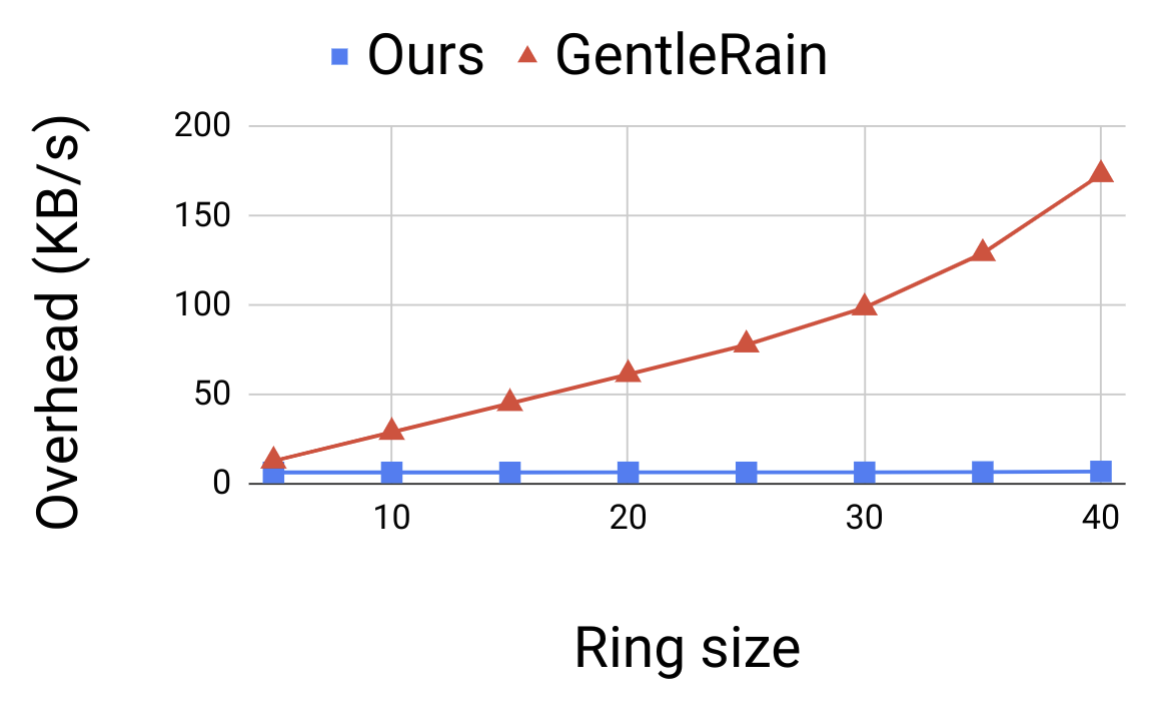}
		\caption{Different Network Delays}
		\label{fig:overhead}
	\end{floatingfigure}
	
	Next, we measure the visibility latency of our algorithm and GentleRain, under the influence of several parameters  including {\em  heartbeat frequency, stabilization frequency, clock skew, update throughput,  ring size and network delay (the last three are presented in Appendix \ref{sec:experiment} due to lack of space).}
	The visibility latency presented in this section is computed as the {\em average latencies over all the updates from all servers}.
	
	\paragraph*{\bf Stabilization Frequencies and Heartbeat Frequencies}
	In this section, we set both stabilization frequencies and heartbeat frequencies to be variables. The network delay is set to be $100ms$ in this experiment.

	\begin{table}[htp]
		\centering
		\begin{tabular}{|l||*{6}{c|}}
			\hline
			\multicolumn{2}{|l|}{\backslashbox{HB fq.\kern-2em}{\kern-1em Stab fq.}}
			&\makebox[3em]{1}&\makebox[3em]{10}&\makebox[3em]{100}
			&\makebox[3em]{500}&\makebox[3em]{1000}\\\hline\hline
			\multirow{4}{*}{Ours} 
			&1                             & 508.09 & 54.87 & 10.44 & 5.69 & 4.87 \\ \cline{2-7} 
			&10                            & 505.92 & 55.53 & 9.37  & 5.88 & 4.76 \\ \cline{2-7} 
			&50                            & 505.33 & 54.52 & 10.02 & 5.47 & 6.21 \\ \cline{2-7} 
			&100                           & 506.51 & 54.13 & 9.22  & 5.09 & 4.75 \\ \cline{2-7} 
			&200                           & 505.77 & 53.28 & 8.47  & 4.64 & 3.97 \\ \hline
			\multirow{4}{*}{GR} 
			&1                             & 1468.42 & 778.51  & 729.88  & 715.95  & 720.42  \\ \cline{2-7} 
			&10                            & 578.99  & 127.41  & 79.89   & 79.82   & 77.02   \\ \cline{2-7} 
			&50                            & 515.96  & 64.95   & 22.63   & 18.23   & 15.71   \\ \cline{2-7} 
			&100                           & 690.01  & 214.11  & 258.27  & 276.5   & 292.56  \\ \cline{2-7} 
			&200                           & 2973.86 & 3612.18 & 2736.22 & 2737.07 & 3685.53 \\ \hline
		\end{tabular}
		\caption{Different Stabilization/Heartbeat Frequency}
		\label{tab:freq}
		\vspace{-5mm}
	\end{table}
	From Table \ref{tab:freq} we can observe that there are significant improvements on latencies by our algorithm comparing to  GentleRain in the simulation. Here are some observations:
	\begin{itemize}
		\item For both algorithms, the visibility latency decreases significantly with higher stabilization frequencies, except the case when the heartbeat frequency is too high in GentleRain. In the latter case, the machine is already overwhelmed by heartbeat message, so increasing stabilization frequency actually damages the performance.
		\item The heartbeat frequency does not influence the visibility latency of our algorithm much, since update messages at a frequency about $5k/\sec$ also carries clock values, and $GST$ computation can proceed with such clock values. However, this is not the case for GentleRain, since each node needs to receive clocks from all other nodes, but the update messages each node receives only come from its neighbors. Then low heartbeat frequencies will delay the $GST$ computation and thus increase the visibility latencies of GentleRain.
		Therefore, the visibility latencies improve with higher heartbeat frequencies in GentleRain, until the number of heartbeat messages is too large for the simulation. Our algorithm does not suffer from such a problem since the heartbeat messages in our algorithm will only be sent to a small set of nodes.
	\end{itemize}

	\paragraph*{\bf Clock Skew}
	To evaluate the influence of clock skew on the visibility latency, we manually add clock skews between any pair of neighbors in the ring. Label the nodes in the ring with id $0,1,\cdots, n-1$ where $n$ is the ring size. For a skew value $t$, we add clock skew $(i\cdot t)/(n-1)$ to node $i$.
	We vary the skew value from $0$ms to $100$ms, and plot the visibility latency change in Figure \ref{fig:skew} below. 
	
	\begin{figure}[htp]
		\begin{subfigure}[b]{0.49\columnwidth}
			\includegraphics[width=\linewidth]{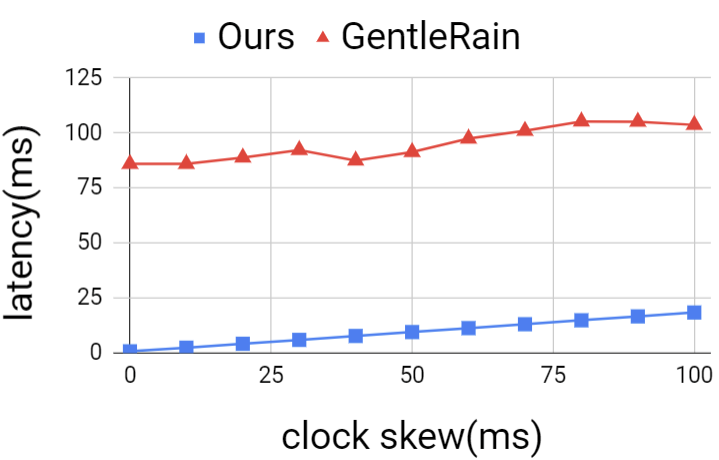}
			\caption{Network Delay = 0ms}
			\label{fig:skew_0}
		\end{subfigure}
		\hfill 
		\begin{subfigure}[b]{0.49\columnwidth}
			\includegraphics[width=\linewidth]{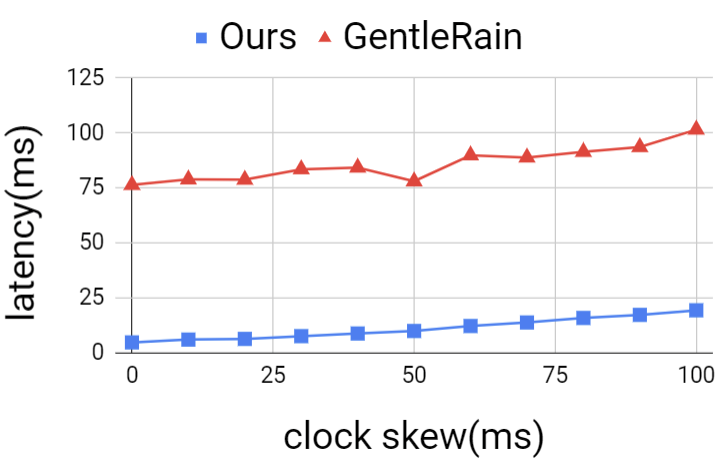}
			\caption{Network Delay = 100ms}
			\label{fig:skew_delay}
		\end{subfigure}
		\caption{Different Clock Skew}
		\label{fig:skew}
	\end{figure}
	
	As we can observe from Figures \ref{fig:skew_0} and \ref{fig:skew_delay}, the remote update visibility latencies increase with the clock skew in both cases. This is predictable since the latency is determined by the minimum clock value received by the server, which is affected by the clock skew between servers.
	Also, our algorithm performs significantly better than GentleRain regarding visibility latency under various clock skews in the simulation.
	
	More simulation results can be found in Appendix \ref{sec:experiment}.
	\section{Discussions and Extensions}\label{sec:discuss}
	
	\subsection{Fault Tolerance}
	In this section, we discuss how various failures such as server failure, network failure or network partitioning may affect our algorithm. 
	Our discussion is analogous to the one in GentleRain \cite{Du2014GentleRainCA}, and can be applied to other stabilization based algorithms as well. 
	
	The main observation is that our stabilization algorithm will guarantee causal consistency even if the system suffers from machine failure, machine slowdown, network delay or partitioning. Recall that in our algorithm, versions are totally ordered by their timestamps which equals the physical time point when the version is created. When a client issues a GET operation, the version returned will have timestamp value no more than the Global Stable Time. 
	
	When a server fails, the client may not receive any response from the server. However, since our algorithm allows clients to migrate across servers, the client can timeout after a period of waiting and then connect to another server to issue operations. The failure of the server will affect the computation of $GST$ at other servers, since the failed server no long sends heartbeat messages to other servers and thus the value of $GST$ at some server may stop updating. In this case, the causal consistency is ensured, since the version returning to the client may be out-of-date but still causally consistent. To make sure the system can make progress and have newer versions visible to the client eventually, other servers should be able to detect the failure eventually. For instance, servers can set a timeout for heartbeat and HS exchanges. If one server does not receive the message from another server after the timeout, it can mark this server as failed. How to recompute the new $GST$ to make progress after failure while ensuring causally consistency is an interesting open problem.
	
	For other issues such as machine slowdown, network delay or partitioning, similarly, the computation of $GST$ may stop making progress, but the version returned to the client is guaranteed to be causally consistent. Then when the failure is recovered, the pending heartbeats or updates can be applied at corresponding servers, and $GST$ can continue to increment.
	One possible failure that can cause the violation of causal consistency is {\em packet loss}, in particular, the loss of update messages. 
	Update loss may result in returning a version to the client that is not causally consistent due to missing dependencies. 
	In practice, we can use reliable communication protocols for transmitting update messages to handle the issue.

	\subsection{Using Hybrid Logical Clocks}\label{sec:hybrid_clock}
	
	To reduce the latency of the PUT operation caused by clock skew, we can use hybrid logical clocks (HLC) \cite{roohitavaf2017causalspartan} instead of a single scalar as the timestamps. The HLC for an event $e$ has two parts, a physical clock $l.e$ and a bounded logical clock $c.e$. The HLC is designed to have the property that if event $e$ happens before event $f$, then $(l.e<l.f)\vee ((l.e=l.f)\wedge(c.e<c.f))$ \cite{roohitavaf2017causalspartan}. By replacing the scalar timestamp with HLC, we may be able to avoid the blocking at line $8$ of Algorithm \ref{alg:server}. More details about HLC can be found in \cite{roohitavaf2017causalspartan}.

	\subsection{Dynamic Systems}\label{sec:dynamic}
	This section will briefly discuss the ideas on how the algorithm can be adapted for dynamic systems where keys can be inserted or deleted, and servers themselves can also be added or removed.
	The change in the system can be essentially modeled as augmented share graph change from $G$ to $G'$. 
	
	When the system experiences changes, the algorithm should guarantee that the causal consistency is not violated. That is, the versions returned to the client should always be causally consistent. Therefore, the algorithm should ensure that during the dynamic change, the Global Stable Time computed is nondecreasing. However, due to the change of the augmented share graph, it is possible that $GST$ computed in the new augmented share graph becomes smaller. To ensure causal consistency, the algorithm can continue to use the old $GST$ value $v$ at the time point when the augmented share graph changes, until the new $GST$ value exceeds $v$. Then the $GST$ used for GET operations is nondecreasing, and the version returned to the client is causally consistent.
	How to design an efficient algorithm for achieving causal consistency in dynamic systems is interesting and left for future work.

	\section{Other Related Work}
	Aside from the previous work mentioned in Section \ref{sec:intro},
	there has been other work dedicated to implementing causal consistency {\em without} any false dependencies in partially replicated distributed shared memory.
	H{\'e}lary and Milani \cite{Hlary2006AboutTE} identified the difficulty of implementing causal consistency in partially replicated distributed storage systems. 
	They proposed the notion of share graph and argued that the metadata size would be large if causal consistency is achieved without false dependencies.
	Reynal and Ahamad \cite{raynal1998exploiting} proposed an algorithm that uses metadata of size $O(mn)$ in the worst case, where $n$ is the number of servers and $m$ is the number of objects replicated.
	Shen et al. \cite{Shen2015CausalCF} proposed two algorithms, {\em Full-Track} and {\em Opt-Track}, that keep track of dependent updates explicitly to achieve causal consistency without false dependencies, where {\em Opt-Track} is proved to be optimal with respect to the size of metadata in local logs and on update messages. Their amortized message size complexity increases linearly with the number of operations, the number of nodes in the system, and the replication factor. 
	Xiang and Vaidya \cite{xiang2017lower} investigated how metadata is affected by data replication and client migration, by proposing an algorithm that utilizes vector timestamps and studying the lower bounds on the size of metadata. 
	The vector timestamp in their algorithm is a function of the share graph and client-server communication pattern, and have worst case timestamp size $O(n^2)$ where $n$ is the number of nodes in the system.
	In the above-mentioned algorithms, in order to eliminate false dependencies, the metadata sizes are large, in particular,  superlinear in the number of servers. 
	In comparison, the global stabilization technique used in our algorithm adopted for partial replication only requires metadata of {\em constant} size, independent of the number of servers, clients or keys.
	
	\section{Conclusion}
	This paper proposes global stabilization for implementing causal consistency in partially replicated distributed storage systems. The algorithm proposed allows each server to store an arbitrary subset of the data, and each client to communicate with an arbitrary set of the servers. We prove the correctness of the algorithm, show the optimality of our Global Stable Time computation under general partial replication, and also discuss several optimizations that can further improve the performance of the algorithm in practice. Simulartion results demonstrate the effectiveness of our $GST$ computation compared to GentleRain for causally consistent partial replication.

	\bibliographystyle{IEEEtran}
	\bibliography{IEEEabrv,IEEEexample}

	\newpage

	\appendices

	\section{Proof for Lemma \ref{lem:1}}\label{sec:lemma1proof}
	
	\begin{proof}
		If two PUTs are issued by the same client, when PUT($k,K$) is issued, by lines $8,10$ of Algorithm \ref{alg:server}, $u_{K}.ut$ will be larger than the client's $\max(PT_c, GT_c)$ value, which is $\geq u_{K'}.ut$ by line $14$ of Algorithm \ref{alg:server} and lines $9,11$ of Algorithm \ref{alg:client}. Hence $u_{K'}.ut<u_{K}.ut$.
		
		If two PUTs are issued by different clients, and the happen-before relation is due to the second client reading the version of the first client's PUT($k',K'$), and then issuing PUT($k,K$). By line $6$ of Algorithm \ref{alg:server} and line $5$ of Algorithm \ref{alg:client}, when the second client issues PUT($k,K$), the dependency timestamp $\max(PT_c, GT_c)$ in line $9$ of Algorithm \ref{alg:client} will be $\geq u_{K'}.ut$. Similarly, by lines $8,10$ of Algorithm \ref{alg:server}, $u_{K}.ut$ will be larger than the client's $\max(PT_c, GT_c)$ value. Hence $u_{K'}.ut<u_K.ut$.
		
		For other cases when PUT($k',K'$) $\rightarrow$ PUT($k,K$), by transitivity we have $u_{K'}.ut<u_{K}.ut$.
		
		Since the timestamp of a version $K$ equals the timestamp for the corresponding replication update $u_K$, we also have
		${K'}.ut<K.ut$.
	\end{proof}
	
	\section{Proof for Lemma \ref{lem:2}}\label{sec:lemma2proof}
	First we list several observations regarding the definitions of the set $L_i(k), R_i(g)$ mentioned in Section \ref{sec:compute_gst}. The observations will be used in later proofs.
	
	\underline{Observation 1:} For any $(v,i),(v',i)\in L_i(k)$ and $k\in \mathcal{K}_{vi}, k'\in \mathcal{K}_{v'i}$, we have $L_i(k)=L_i(k')$.

	\underline{Observation 2:} For any $(v,i)\in L_i(k)$, if $(v,i)\in R_j(g)$ for some server $j\neq i$, we have $L_i(k)\subseteq R_j(g)$.

	\underline{Observation 3:} For a server set $g$ containing server $i,j$ and $(v,i)\in L_i(k)$, if $(v,i)\in R_j(g)$, we have $LD_i(k)=HS_i(g)$.
	
	\begin{proof}[Proof of the lemma]
		
		In order to have $K\dep K'$, there must be a chain of versions on a simple path (no vertex repetition) from $i'$ to $i$ in $G^a$ such that $K= K_1\dep K_2\dep \cdots\dep K_m\dep K_{m+1}=K'$ where  each version $K_x$ corresponds to key $k_x$.
		
		We prove the lemma in two cases, $i'=i$ and $i'\neq i$.

		\textbf{Case I: $i'=i$. }Since the version $K$ could be due to a local PUT at server $i$ or a non-local PUT at a server other than $i$, there are two cases.

		\begin{enumerate}
			\item $K$ is due to a non-local PUT at a server other than $i$.  There are two cases, namely none of $K_x$ is issued at $i$ for $1\leq x\leq m+1$, or at least one $K_x$ is issued at $i$.
			\begin{enumerate}
				\item None of $K_x$ is issued at $i$. This implies that there exists a simple cycle $C=(i,v_1,\cdots, v_m,i)$ such that $k\in \K_{iv_1}$, $k'\in \K_{iv_m}$, and $K$ is the result of PUT($k,K$) at $v_1$, $K'$ is the result of PUT($k',K'$) at $v_m$. Since $K\dep K'$, the dependency is propagated along the path $v_m, v_{m-1}, \cdots, v_1$ in $G^a$. We illustrate one possible execution as follows. 
				
				First, a client $c_{m+1}$ issues PUT($k',K'$) at server $v_m$, which leads to an update $u_{K'}$ from $v_m$ to $i$. 
				Then for $x=m, m-1, \cdots, 2$ sequentially, a client $c_x$ reads the version written by the previous client $c_{x+1}$ from server $v_x$ via a  GET operation at server $v_x$. If $(v_{x-1}, v_x)\in E_1(G^a)$, client $c_x$ then issues PUT($k_x, K_x$) at $v_x$ where $k_x\in \K_{v_{x-1}v_x}$, which leads to an update message from $v_x$ to $v_{x-1}$. If $(v_{x-1}, v_x)\in E_2(G^a)$, without loss of generality, suppose $c_x$ can access both $v_{x-1},v_x$. Then $c_x$ issues PUT($k_x, K_x$) at $v_{x-1}$ where $k_x\in \K_{v_{x-1}}$. In the end, client $c_1$ read the version $K_2$,  written by client $c_2$, from server $v_1$, and issues PUT($k,K$) at server $v_1$, which results in an update $u_K$ from $v_1$ to $i$.
				By the definition of happens-before relation, it is clear that PUT($k',K'$) $\rightarrow$ PUT($k,K$), namely $K\dep K'$.

				\begin{figure}[H]
					\centering
					\includegraphics[width=0.8\linewidth]{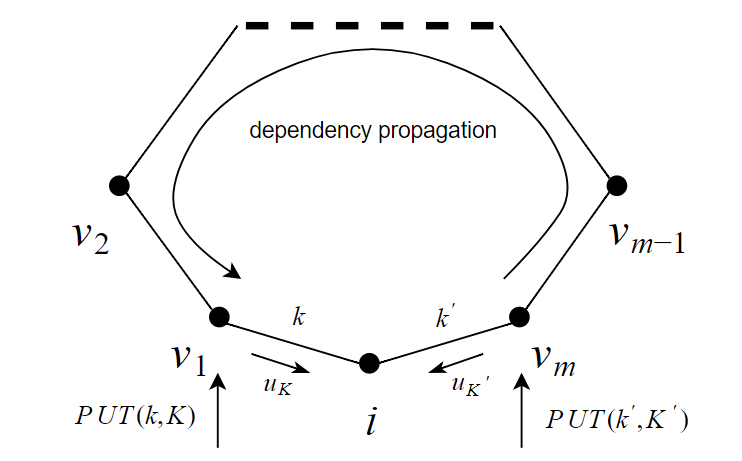}
					\caption{Illustration for Case I.1(a)}
					\label{fig:case11a}
				\end{figure}
				
				We first prove that $K'$ is received by server $i$. Let $HB^0_{v_mi}$ denote the heartbeat value received by $i$ from $v_m$ when $K$ is read by the client. Since $K$ is read by the client, by line $5$ of Algorithm \ref{alg:server} we have $K.ut\leq GST$. By definition $GST=\min(LD_i(k), \max(RD_i(g), rd))\leq LD_i(k)$, we have $K.ut\leq LD_i(k)$. By the definition of set $L_i(k)$, we have $(v_m,i)\in L_i(k)$, and thus $LD_i(k)=\min_{(v,i)\in L_i(k)}(HB_{vi})\leq HB^0_{v_m i}$, which implies that $K.ut\leq HB^0_{v_m i}$. By Lemma \ref{lem:1}, $K'.ut<K.ut$ since $K\dep K'$. Therefore we have $K'.ut\leq HB^0_{v_m i}$, which implies that $K'$ is received by server $i$ since the channel is assumed to be FIFO.
				
				Now we prove that $K'$ is visible to client $c$ from server $i$.
				Let $GST^0$ denote the Global Stable Time when $K$ is read by the client, then $GST^0\geq K.ut$ by line $5$ of Algorithm \ref{alg:server}.
				Since $(v_1,i),(v_m,i)\in L_i(k)$, by Observation $1$, $L_i(k)=L_i(k')$ and thus $LD_i(k)=LD_i(k')$. Notice that at any server, the heartbeat values received from another server is nondecreasing, thus the value of $LD_i(k')$ and $RD_i(g)$ at any server are also nondecreasing. By line $6$ and $2$ of Algorithm \ref{alg:client},  the value of $rd$ computed at line $2$ of Algorithm \ref{alg:client} is also nondecreasing. Therefore when client $c$ issues $GET(k')$ at server $i$, $GST=\min(LD_i(k'), \max(RD_i(g), rd))\geq GST^0 \geq K.ut$.  By Lemma \ref{lem:1}, $K'.ut<K.ut$, which implies that $GST\geq K'.ut$ and thus $K'$ is visible to client $c$ from server $i$.

				\item At least one $K_x$ is issued at $i$. Let $K_f$ be the first version that is issued at $i$, namely $K_f$ is the version issued at $i$ with the largest subscript. 
				Since $K_f\dep K_{f+1}\dep\cdots\dep K'$, 
				there exists a simple cycle $C=(i,v_{f+1}, v_{f+2}, \cdots, v_m,i)$, where $k'\in \K_{iv_m}$ and $K'$ is the result of PUT($k',K'$) at $v_m$.
				Depending on the edge $(i,v_{f+1})$ and how dependencies propagate, there are two cases.
				
				\begin{enumerate}
					\item 
					$(i,v_{f+1})$ is a real edge. Let $k_{f+1}\in \K_{iv_{f+1}}$ and $K_{f+1}$ is the result of PUT($k_{f+1},K_{f+1}$) at $v_{f+1}$.
					The dependency between $K'$ and $K_{f+1}$ is propagated along the path $(i,v_m, \cdots, v_{f+1})$ similarly as in Case I.1(a), and $K_f$ is issued by some client $c'$ after $c'$ read $K_{f+1}$ from server $i$. Then when $K_{f+1}$ is read by the client $c'$ at server $i$, the conclusion of Case I.1(a) guarantees that the lemma holds.
					\begin{figure}[H]
						\centering
						\includegraphics[width=0.7\linewidth]{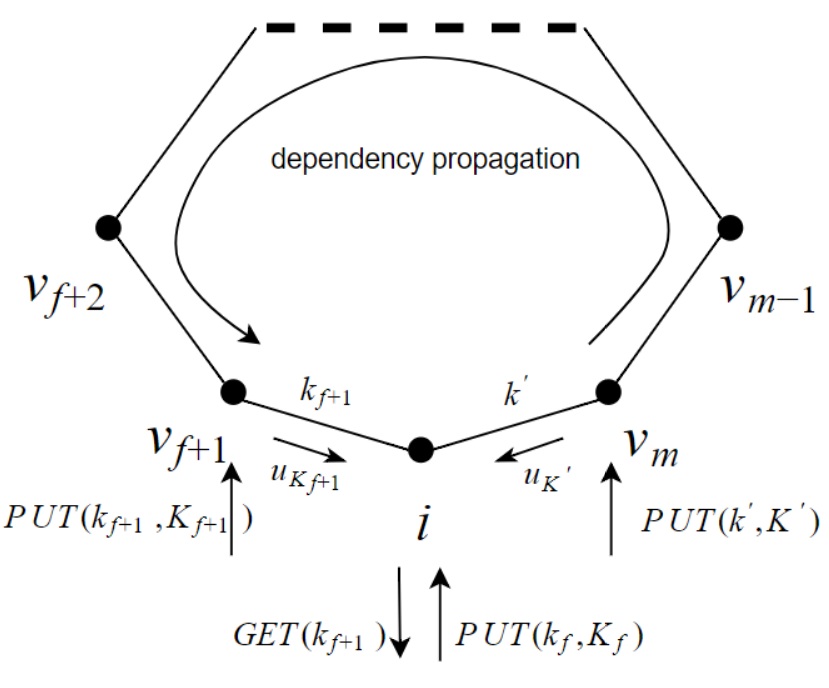}
						\caption{Illustration for Case I.1(b).i}
					\end{figure}
					
					\item 
					$(i,v_{f+1})$ is a virtual edge. Without loss of generality, suppose that $i,v_{f+1}\in S_{c'}$.
					The dependency between $K_{f}$ and $K'$ is propagated along the path similarly as in Case I.1(a), and $K_{f}$ is issued by client $c'$ at server $i$ after $c'$ reads $K_{f+2}$ from server $v_{f+1}$.

					We first prove that $K'$ is received by server $i$. Let $HB^0_{v_mi}$ denote the heartbeat value received by $i$ from $v_m$ when $K_{f+2}$ is read by the client from server $v_{f+1}$. Consider the time point when $K_{f+2}$ is read by the client from server $v_{f+1}$. By line $5$ of Algorithm \ref{alg:server} we have $K_{f+2}.ut\leq GST$.  By definition, $RD_{v_{f+1}}(g)=\min_{(x,y)\in R_{v_{f+1}}(g)}(HB_{xy})\leq HB^0_{v_mi}$ since $(v_m, i)\in R_{v_{f+1}}(g)$. Also, by line $2$ and $6$ of Algorithm \ref{alg:client}, $rd=\min_{j\in S_c, j\neq v_{f+1}}HS_c[j]\leq HS_c[i]\leq HB^0_{v_mi}$. 
					When $K_{f+2}$ is returned, by the definition of $GST$,  $GST=\min\left( LD_{v_{f+1}}(k_{f+2}), \max(RD_{v_{f+1}}(g), rd)\right)\leq \max(RD_{v_{f+1}}(g), rd)\leq HB^0_{v_mi}$. 
					Hence we have $K_{f+2}.ut\leq GST\leq HB^0_{v_mi}$. By Lemma \ref{lem:1}, $K'.ut<K_{f+2}.ut$ since $K_{f+2}\dep K'$.
					Therefore we have $K'.ut\leq HB^0_{v_m i}$, which implies that $K'$ is received by server $i$ since the channel is assumed to be FIFO.
					
					\begin{figure}[H]
						\centering
						\includegraphics[width=0.7\linewidth]{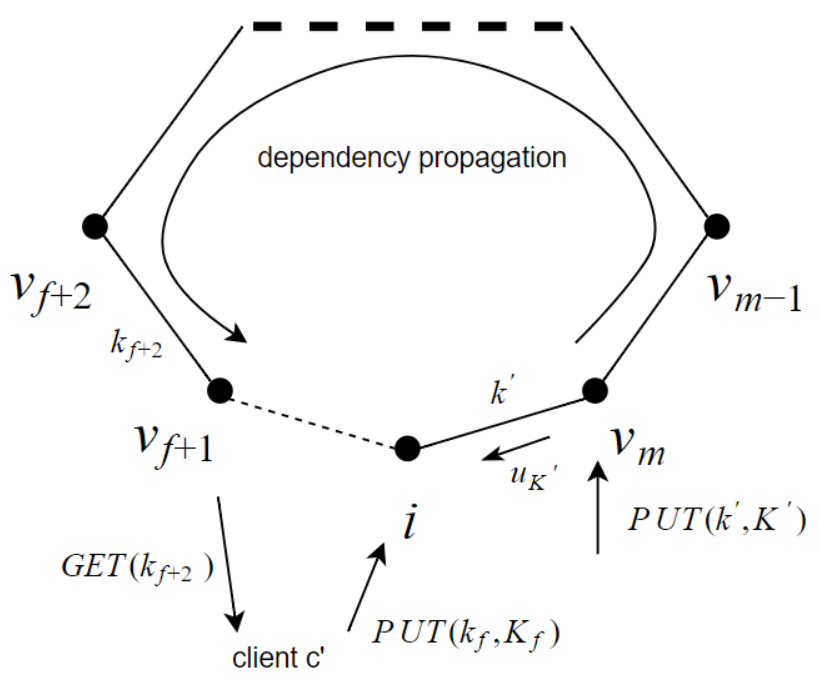}
						\caption{Illustration for Case I.1(b).ii}
					\end{figure}
					
					Now we prove  $K'$ is visible to client $c$ from server $i$. 
					
					We first show that $LD_i(k')\geq K'.ut$ when client $c$ issues $GET(k')$ to server $i$.
					Consider the time point when $K_{f+2}$ is read by the client $c'$ from server $v_{f+1}$. We have $K_{f+2}.ut\leq GST\leq \max(RD_{v_{f+1}}(g'), rd)$ where $g'=S_{c'}$. Notice that $\forall (v,i)\in L_i(k')$, we have $(v,i)\in R_{v_{f+1}}(g')$, since we can find a cycle containing $(v,i)$ that satisfies the requirement for $R_{v_{f+1}}(g')$. This implies that $LD_i(k')\geq RD_{v_{f+1}}(g')$ at any time point. For the value of $rd$, it is computed as $rd=\min_{j\in S_{c'}, j\neq v_{f+1}} HS_{c'}[j]\leq HS_{c'}[i]$. By definition, $HS_{c'}[i]\leq HS_i(g')\leq LD_i(k')$. The first inequality is because that $HS_{c'}[i]$ is updated by $HS_i(g')$, and the second inequality is because that $HS_i(g')$ includes the heartbeat value $HB_{vi}$ for all $(v,i)\in  L_i(k')$ and calculates the minimum.
					Therefore, we have $rd\leq LD_i(k')$, together with $RD_{v_{f+1}}(g')\leq LD_i(k')$ and $K_{f+2}.ut\leq \max(RD_{v_{f+1}}(g'), rd)$, we have $K_{f+2}.ut\leq LD_i(k')$ at the time point when $K_{f+2}$ is returned. By Lemma \ref{lem:1}, $K'.ut<K_{f+2}.ut$ and thus $K'.ut\leq LD_i(k')$. Since $LD_i(k')$ is nondecreasing, this condition remains true later when client $c$ reads $K'$ from $i$.
					
					Now we show that $\max(RD_i(g), rd)\geq K'.ut$  when client $c$ issues $GET(k')$ to server $i$. 
					When $K$ is read by the client $c$ from server $i$, by line $5$ of Algorithm \ref{alg:server} we have $K.ut\leq GST\leq \max(RD_i(g), rd)$. Since the value of $\max(RD_i(g), rd)$ is nondecreasing, when client $c$ issues $GET(k')$ later, we also have $K.ut\leq \max(RD_i(g), rd)$. By Lemma \ref{lem:1}, $K'.ut<K.ut$ and thus $K'.ut\leq \max(RD_i(g), rd)$.
					
					Summarizing the conclusions above, we have 
					$GST=\min(LD_{i}(k'), \max(RD(g),rd))\geq K'.ut$, which implies that $K'$ is visible to client $c$ from server $i$.

				\end{enumerate}
			\end{enumerate}

			\item $K$ is due to a local PUT at server $i$. 
			Since $K$ is issued at server $i$, Case I.1(b) proves that the lemma holds.
			
		\end{enumerate}

		\textbf{Case II: $i'\neq i$. }
		
		\begin{enumerate}
			\item First consider the case where there exists at least one $K_x$ issued at server $i'$. Let $K_f$ be the last version that is issued at server $i'$, namely $K_f$ is the version with the largest subscript. 
			Then the same proof for Case I.1(b) proves that $K'$ is received by server $i'$, and $LD_{i'}(k')\geq K'.ut$. 
			
			Now we will prove that $K'$ is visible to client $c$ from server $i'$. When $K$ is read by client $c$ from server $i$, by line $5$ of Algorithm \ref{alg:server}, we have $K.ut\leq GST=\min\left( LD_i(k), \max(RD_i(g), rd)\right)\leq \max(RD_i(g), rd)$  where $g=S_c$. By definition, $RD_i(g)=\min_{j\in g, j\neq i}(HS_j(g))$ and $rd=\min_{j\in g,j\neq i}(HS_c[j])$. Since the client will store the largest $HS$ values for each server $j\in S_c$, we have $HS_c[j]\geq K.ut>K'.ut$ stored at the client $c$ for each server $j\neq i$ in $S_c$. 
			
			Now we will show that $HS_c[i]\geq K'.ut$ when client $c$ issues $GET(k')$ to server $i'$.  Since $K=K_1 \dep K_2\dep \cdots \dep K_f$, there exists a simple path $(i,v_1,\cdots, v_m,i')$ connects $i'$ and $i$ that propagates the dependency above. Similarly to Case I.1.(b), there are two cases, i.e. $(i,v_1)$ is a real edge or virtual edge. 
			If $(i,v_1)$ is a real edge, let version $K_t$ of key $k_t$ be the version that is sent from $v_1$ to $i$, and read by some client at $i$. Since $K_t$ is visible, we have $LD_i(k_t)\geq GST\geq K_t.ut$. Notice that $(v_1,i)\in R_{i'}(g)$ due to the simple path above,  by Observation 3, we know that $LD_i(k_t)=HS_i(g)$. Thus $HS_i(g)\geq K_t.ut>K'.ut$.
			If $(i,v_1)$ is a virtual edge, let client $c'$ be the one that gets a version $K_t$ from server $v_1$ and then puts a version to server $i$. When $K_t$ is returned, we have $HS_i(S_{c'})\geq K_t.ut$. Notice that for $\forall (u,i)\in R_{i'}(g)$ where $g=S_c$, we also have $(u,i)\in R_{v_1}(S_{c'})$ since $v_1,i'$ are connected by a simple path. Thus $HS_i(g)\geq HS_i(S_{c'})\geq K_t.ut>K'.ut$.
			Since the client will keep largest $HS$ values, we have $HS_c[i]\geq HS_i(g)\geq K'.ut$.
			
			Then, when client $c$ issues $GET(k')$ to server $i'$, we have proved that  $LD_{i'}(k')\geq K'.ut$, $HS_c[j]\geq K'.ut$ stored at the client $c$ for each server $j\in S_c$. According to line $2$ of Algorithm \ref{alg:client}, the dependency clock value that client $c$ passes to server $i'$ is $rd=\min_{j\in S_c, j\neq i'} HS_c[j]\geq K'.ut$. 
			Recall that we already proved $LD_{i'}(k')\geq K'.ut$.
			Then $GST=\min\left( LD_{i'}(k), \max(RD_{i'}(g), rd)\right)\geq \min( LD_{i'}(k), rd)\geq K'.ut$, and hence $K'$ is visible to client $c$ from server $i'$.

			\item Now consider the case where none of $K_x$ is issued at $i'$.
			Then there exists a simple path $(i',v_m,\cdots, v_1, i)$ such that the causal dependencies are propagated through the path.
			Notice that the situation is identical to the second part of Case II.1 above, and the same proof will show that $K'$ is received by server $i'$, and $K'$ is visible to client $c$ from server $i'$. 
			
		\end{enumerate}

	\end{proof}
	
	\section{Proof for Theorem \ref{thm:1}}\label{sec:thmproof}
	\begin{proof}
		\textbf{To prove the first condition}, which is:
		Let $k$ and $k'$ be any two keys in the store. Let $K$ be a version of key $k$, and $K'$ be a version of key $k'$ such that $K \dep K'$. 
		For any client $c$ that can access both $k$ and $k'$, when $K$ is read by client $c$, $K'$ is visible to $c$.
		
		If $K'$ is due to a local PUT at the server that client $c$ is accessing, then by line $5$ of Algorithm \ref{alg:server}, $K'$ is visible to client $c$.
		Otherwise, if $K'$ is due to a non-local PUT, according to Lemma \ref{lem:2}, $K'$ is received by the server which the client is accessing, and is also visible to the client.
		
		\textbf{To prove the second condition}, which is:
		A version $K$ of a key $k$ is visible to a client $c$ after $c$ completes PUT($k,K$) operation.
		
		Consider a client $c$ issuing GET($k$) after a PUT($k,K$) operation. If client $c$ reads from the same server, according to line $5$ of Algorithm \ref{alg:server}, $K$ is visible to the client. If client $c$ reads from a different server, to pass lines $3,4$ of Algorithm \ref{alg:server}, we have $K.ut\leq PT_c = t\leq GST$. By definition, $GST=\min(LD_i(k), \max(RD_i(g), rd))\leq LD_i(k)$. Thus $K.ut\leq LD_i(k)$, and the definition of $LD_i(k)$ implies that $K$ is already received by $i$. Then, since $K.ut\leq GST$, version $K$ is visible to client $c$.
	\end{proof}
	
	\section{Proof for Theorem \ref{thm:2}}\label{sec:optproof}
	\begin{proof}
		Recall the definition of $GST$ from Section \ref{sec:compute_gst}.

		\[
		GST= \min\left( LD_i(k), \max(RD_i(g), rd)\right)
		\]
		where
		$LD_i(k)=\min_{(v,i)\in L_i(k)}(HB_{vi})$,
		$RD_i(g)=\min_{(x,y)\in R_i(g)}(HB_{xy})$,
		and 
		$rd=\min_{j\in S_c, j\neq i} (HS_c[j])$.
		
		By line $6$ of Algorithm \ref{alg:client} and line $6$ of Algorithm \ref{alg:server}, the value of $HS_c[j]$ the client keeps is the largest $HS_j(g)$ value it has seen so far from servers it accessed so far for $\forall j\in S_c$. By definition, 
		$HS_j(g)=\min_{\forall (z,j)\in R_i(g)}(HB_{zj})$, which implies that $rd$ is also computed as the minimum value of a set of heartbeat values.
		
		By the definitions above, we observe that our $GST$ is computed as the minimum of a set of heartbeat values from server $x$ to server $y$ where  $(x,y)\in L_i(k)\cup R_i(g)$.
		Let $HB_{pq}$ be the minimum heartbeat value from the set and therefore $GST=HB_{pq}$. 
		There are two cases.
		
		\textbf{Case I:} $(p,q)\in L_i(k)$, and thus $q=i$.

		By the definition of $L_i(k)$, there exists a simple cycle $(i,v_1,\cdots,v_m, i)$ of length $\geq2$ in $G^a$ such that $m\geq 1$, $k\in (v_1,i)$, we have $(v_1, i)\in L_i(k)$, and $(v_m,i)\in L_i(k)$ if $(v_m,i)$ is a real edge. First observe that due to the fact that version $K$ with $K.ut>GST=HB_{pi}$ is returned to the client, we have $p\neq v_1$, otherwise version $K$ is not received by server $i$ yet since the latest heartbeat value received by $i$ from $v_1$ is $HB_{pi}<K.ut$. Without loss of generality, let $p=v_m$. 
		We can show the following possible execution that will violate causal consistency. Let there be a $PUT(k', K')$ at server $p$ which results in  a version $K'$ with timestamp $H_{pi}<K'.ut < K.ut$ such that $K\dep K'$. 
		The causal dependency can be created by the same procedure as described in Case I of the proof for Lemma \ref{lem:2}. 
		For completeness, we state the procedure here again.
		First, a client $c_{m+1}$ issues PUT($k',K'$) at server $j$, which leads to an update $u_{K'}$ from $j$ to $i$. 
		Then for $x=m, m-1, \cdots, 2$ sequentially, a client $c_x$ reads the version written by the previous client $c_{x+1}$ from server $v_x$ via a  GET operation at server $v_x$. If $(v_{x-1}, v_x)\in E_1(G^a)$, client $c_x$ then issues PUT($k_x, K_x$) at $v_x$ where $k_x\in \K_{v_{x-1}v_x}$, which leads to a replication update from $v_x$ to $v_{x-1}$. If $(v_{x-1}, v_x)\in E_2(G^a)$, without loss of generality, suppose $c_x$ can access both $v_{x-1},v_x$. Then $c_x$ issues PUT($k_x, K_x$) at $v_{x-1}$ where $k_x\in \K_{v_{x-1}}$. In the end, client $c_1$ reads the version $K_2$,  written by client $c_2$, from server $v_1$, and issues PUT($k,K$) at server $v_1$, which results in an update $u_K$ from $v_1$ to $i$.
		By the definition of happens-before relation, it is clear that PUT($k',K'$) $\rightarrow$ PUT($k,K$), namely $K\dep K'$.

		\begin{figure}[H]
			\centering
			\includegraphics[width=0.7\linewidth]{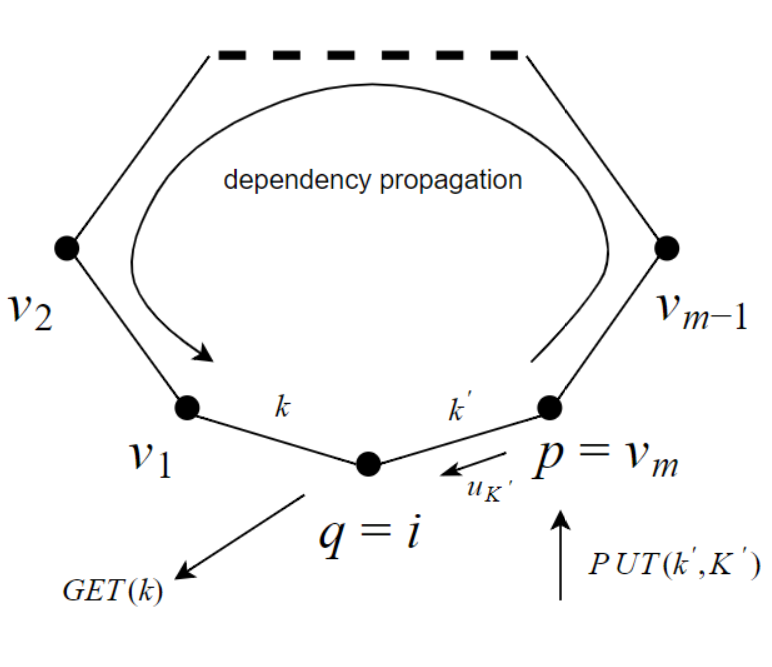}
			\caption{Illustration for Case I}
		\end{figure}
		
		Since $K'.ut >H_{pi}$, $K'$ is not received by $i$ at the time when version $K$ is returned to the client $c$. Now let $u_{K'}$ be delayed indefinitely, which is possible since the system is asynchronous. Consider the case that after reading version $K$, client $c$ issues $GET(k')$ at server $i$. Suppose that client $c$ does not issue any $PUT$ operation before, and thus its $PT_c=0$. Notice that the get operation is non-blocking when $PT_c=0$ by lines $3,4$ of Algorithm \ref{alg:server}, it is possible that an older version $K_0'$ of key $k'$ such that $K'\dep K_0'$ is returned to client $c$ since $u_{K'}$ is delayed and not received by server $i$.
		Hence $K'$ is not visible to the client $c$, which violates the causal consistency. 
		
		\textbf{Case II:} $(p,q)\in R_i(g)$. Then by definition, there exists a simple cycle $(i=v_1,\cdots,v_{m-1}=p,v_m=q)$ of length $\geq2$ in $G^a$ such that $m\geq 2$ and $i,q\in g$.
		
		\begin{figure}[h]
			\centering
			\includegraphics[width=0.7\linewidth]{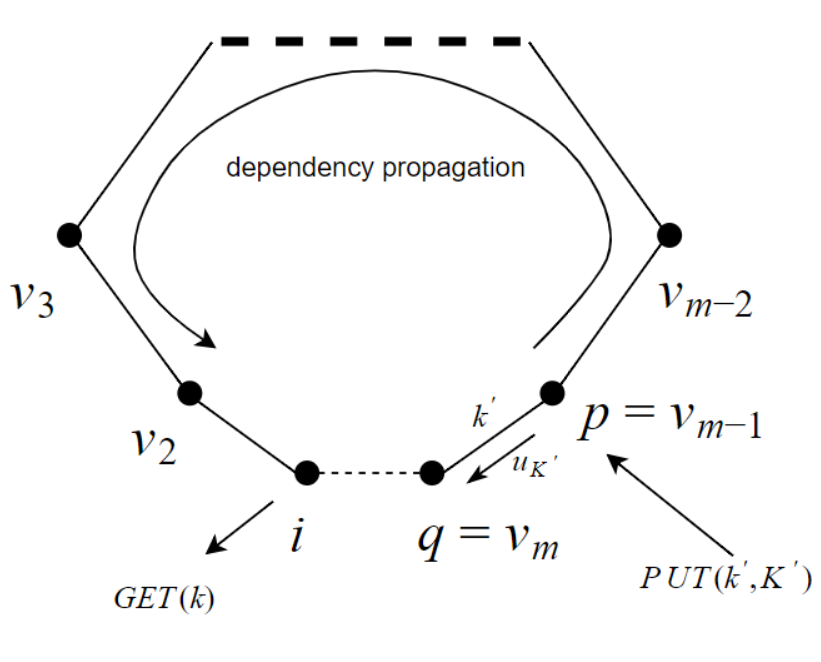}
			\caption{Illustration for Case II}
		\end{figure}
		
		Let $k'\in \mathcal{K}_{pq}$. Let there be a $PUT(k', K')$ at server $p$ which results in  a version $K'$ with timestamp $H_{pq}<K'.ut < K.ut$ such that $K\dep K'$. The causal dependency can be created by a similar procedure as described in Case I above, with differences at the end: client $c_1$ reads the version $K_2$ from server $v_2$, and issues PUT($k_1,K_1$) at server $v_2$ where $k_1\in \mathcal{K}_{iv_2}$. Then some client $c'$ that only access server $i$ ($S_{c'}=\{i\}$) reads the version $K_1$ and issues PUT($k,K$) at server $i$.
		The fact that $S_{c'}=\{i\}$ ensure that when client $c'$ can read $K_1$ without $K'$ being received by $q$.
		
		Now let  $u_{K'}$ be delayed indefinitely. Suppose that after client $c$ gets version $K$, it issues $GET(k')$ at server $i'$. Similar to Case I, $K'$ is not visible to client $c$, which violates the causal consistency.

	\end{proof}
	
	\section{More Simulation Results}\label{sec:experiment}
	\paragraph*{\bf Update Throughput}
	Since we simulate servers by running multiple server processes in a single machine, there is a limitation on the maximum update throughput, which is about $12.5k$ updates per second for each server program when we have $10$ processes running. There also exists a threshold after which the machine cannot handle the update messages in time, leading to a dramatic increase in the visibility latencies. To find such threshold, we plot the latency changes with respect to the update throughput in Figure \ref{fig:ops} and \ref{fig:ops_delay} with $0ms$ and $100ms$ network delays respectively.
	
	\begin{figure}[htp]
		\begin{subfigure}[b]{0.49\columnwidth}
			\includegraphics[width=\linewidth]{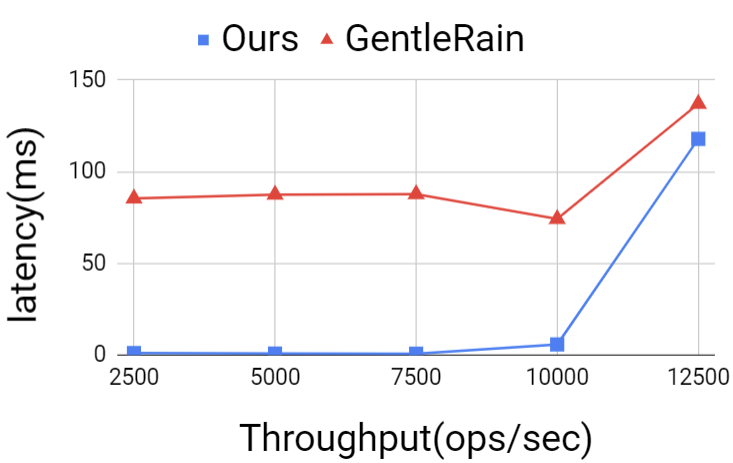}
			\caption{Network Delay = 0ms}
			\label{fig:ops}
		\end{subfigure}
		\hfill 
		\begin{subfigure}[b]{0.49\columnwidth}
			\includegraphics[width=\linewidth]{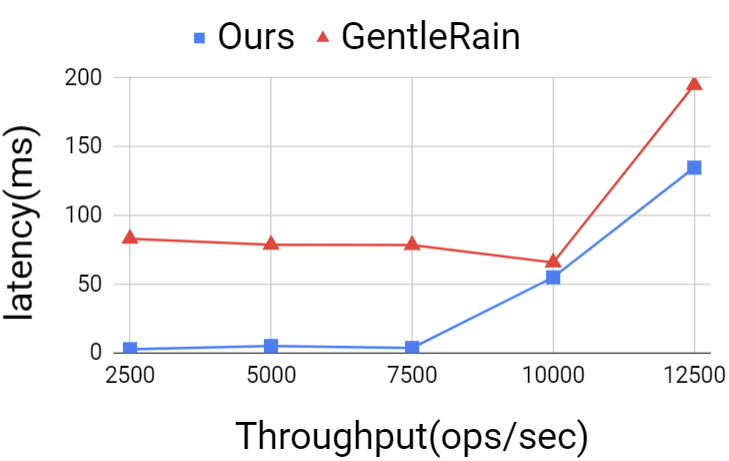}
			\caption{Network Delay = 100ms}
			\label{fig:ops_delay}
		\end{subfigure}
		\caption{Different Update Throughput}
		\label{fig:op}
	\end{figure}

	As we can see from Figure \ref{fig:ops} and \ref{fig:ops_delay}, the threshold would be some value $>10k$ when network delay is $0ms$ and  $>7.5k$ when network delay is $100ms$.  Hence for other evaluations, we set the update throughput to be $5k/\sec$ for each node, since we will increase the other parameters such as ring size, heartbeat frequency, and stabilization frequency for other experiments.

	\paragraph*{\bf Ring Sizes}
	Intuitively, the ring size will affect the visibility latency of the stabilization algorithm in GentleRain, since the number of heartbeat values received by any node will grow linearly with the ring size, leading to smaller $GST$ and larger visibility latencies. However, our algorithm will not be affected too much since the number of heartbeat values received is equal to the number of neighbors in the ring.
	Figure \ref{fig:size_0} and \ref{fig:size_delay} below validate the discussion above, and demonstrate the scalability of our algorithm. In both cases, the visibility latency in our algorithm remains relatively stable while the latency in GentleRain increases as ring size increments. Notice that with network delay of $100ms$, the visibility latency grows dramatically larger (more than $1000ms$) as ring size increases.
	The reason may be that the queue size of messages becomes too large with artificial delay when the ring size is large, which results in high latency in our simulation.
	
	\begin{figure}[htp]
		\begin{subfigure}[b]{0.49\columnwidth}
			\includegraphics[width=\linewidth]{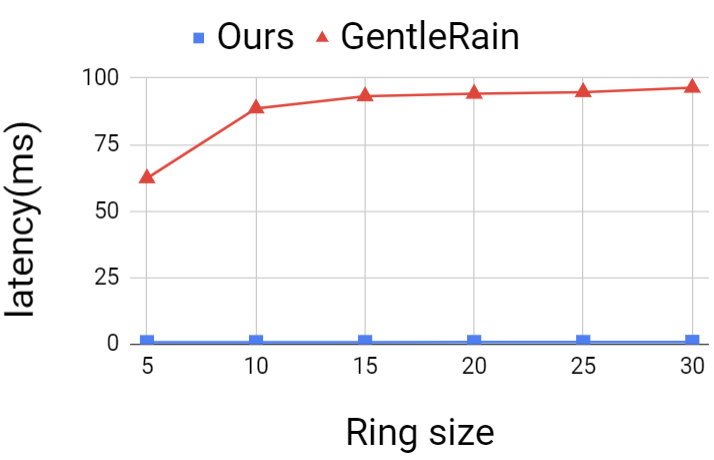}
			\caption{Network Delay = 0ms}
			\label{fig:size_0}
		\end{subfigure}
		\hfill 
		\begin{subfigure}[b]{0.49\columnwidth}
			\includegraphics[width=\linewidth]{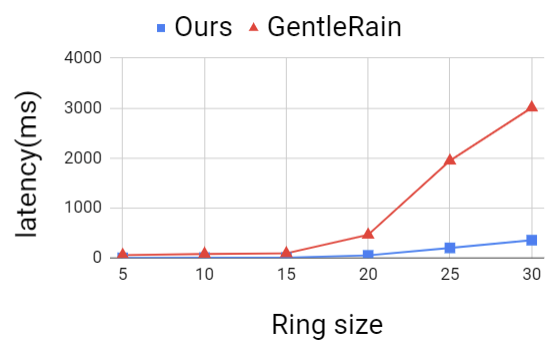}
			\caption{Network Delay = 100ms}
			\label{fig:size_delay}
		\end{subfigure}
		\caption{Different Ring Size}
		\label{fig:size}
	\end{figure}

	{\bf Network Latencies}
	
	To measure the influence of network latencies on the visibility latency, we manually add extra delays to all network packages via Linux {\em tc} command. Although the network delays are set to be constants in our experiment which may not be true in practice, the results give us some insights on how network delay will affect the visibility latencies. As shown in Figure \ref{fig:delay}, the visibility latency is mostly stable with low network delays ($<150ms$), and increases when network delay becomes large ($>150ms$). By definition, visibility latency is the period from when a remote update is received to when the remote update can be returned. Hence in theory, with good network conditions, the visibility latency should not be affected much by network delays. However, when network conditions become worse, 
	the computation of $GST$ may be negatively affected by the network delays, leading to increment in the visibility latencies.

	\begin{figure}[H]
	\centering
	\includegraphics[height=1.25in]{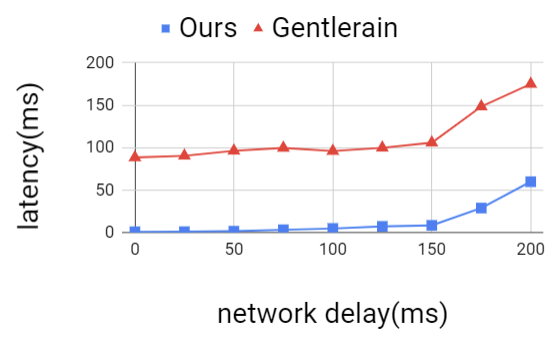}
	\caption{Different Network Delays}
	\label{fig:delay}
	\end{figure}

\end{document}